\pdfoutput=1

\documentclass[format=acmsmall, review=false, screen=true,natbib]{acmart}

\usepackage{booktabs} % For formal tables
\usepackage[skip=0pt]{caption}
\usepackage{soul}

\usepackage{url}  
\usepackage{amsmath}
\usepackage{amsfonts}
\usepackage{amsthm}
\usepackage{dsfont}
\usepackage{graphicx}
\usepackage{epstopdf}
\usepackage{epsfig}
\usepackage{subfigure}
\usepackage{lineno}
\usepackage{multirow}
\usepackage{booktabs}
\usepackage{color}
\usepackage[flushleft]{threeparttable}
\usepackage[oldenum]{paralist}
\usepackage[skip=0pt]{caption}
\usepackage{enumitem}
\usepackage{acronym}
\usepackage{paralist}
\usepackage{algpseudocode}
\usepackage{algorithm}
\usepackage[np, autolanguage]{numprint}
\usepackage{hyperref}

\usepackage[normalem]{ulem}

\graphicspath{{./figures/}}

\DeclareMathOperator*{\argmax}{arg\,max}
\DeclareMathOperator*{\argmin}{arg\,min}

\acrodef{MergeDTS}{Merge Double Thompson Sampling}
\acrodef{MergeRUCB}{Merge Relative Upper Confidence Bound}
\acrodef{UCB}{Upper Confidence Bound}
\acrodef{LCB}{Lower Confidence Bound}
\acrodef{TS}{Thompson Sampling}
\acrodef{DTS}{Double Thompson Sampling}
\acrodef{SBM}{Singleton Bandits Machines}
\acrodef{RUCB}{Relative Upper Confidence Bound}
\acrodef{MAB}{multi-armed bandit}
\acrodef{MDB}{Multi-Dueling Bandit}
\acrodef{DCM}{dependent click model}
\acrodef{RMED1}{Relative Minimum Empirical Divergence}
\acrodef{DMED}{Deterministic Minimum Empirical Divergence}
\acrodef{PI}{Probabilistic Interleaving}
\acrodef{SOSM}{Sample-Only Scored Multileave}
\acrodef{REX3}{Relative Exponential-weight algorithm for Exploration and Exploitation}
\acrodef{EXP3}{Exponential-weight algorithm for Exploration and Exploitation}

\setcopyright{rightsretained}
\acmYear{2020}
\copyrightyear{2020}

\acmJournal{TOIS}

\author{Chang Li}
\affiliation{%
	\institution{University of Amsterdam}
	\streetaddress{Science Park 904}
	\postcode{1098 XH}
	\city{Amsterdam}
	\country{The Netherlands}}
\email{c.li@uva.nl}
\authornote{Corresponding author.}

\author{Ilya Markov}
\affiliation{%
	\institution{University of Amsterdam}
	\streetaddress{Science Park 904}
	\postcode{1098 XH}
	\city{Amsterdam}
	\country{The Netherlands}}
\email{i.markov@uva.nl}

\author{Maarten de Rijke}
\affiliation{%
	\institution{University of Amsterdam}
	\streetaddress{Science Park 904}
	\postcode{1098 XH} 
	\city{Amsterdam}
	\country{The Netherlands}}
\email{m.derijke@uva.nl}
\affiliation{%
	\institution{Ahold Delhaize}
	\streetaddress{Provincialeweg 11}
	\postcode{1506 MA}
	\city{Zaandam}
	\country{The Netherlands}
}	

\author{Masrour Zoghi}
\affiliation{%
	\institution{Microsoft}
	\city{Redmond}
	\state{WA} 
	\country{USA}
 }
\email{masrour@zoghi.org}
\authornote{Now at Google.}

\thanks{This research was partially supported by
the Innovation Center for Artificial Intelligence (ICAI),
and
the Netherlands Organisation for Scientific Research (NWO)
under pro\-ject nr
612.\-001.\-551. % CLEAR
All content represents the opinion of the authors, which is not necessarily shared or endorsed by their respective employers and/or sponsors.}

\title{MergeDTS: A Method for Effective Large-Scale Online Ranker Evaluation}

\begin{document}

\begin{abstract}
Online ranker evaluation is one of the key challenges in information retrieval. 
While the preferences of rankers can be inferred by interleaving methods, the problem of how to effectively choose the ranker pair that generates the interleaved list without degrading the user experience too much is still challenging. 
On the one hand, if two rankers have not been compared enough, the inferred preference can be noisy and inaccurate. 
On the other, if two rankers are compared too many times, the interleaving process inevitably hurts the user experience too much.  
This dilemma is known as the \emph{exploration versus exploitation} tradeoff. 
It is captured by the $K$-armed dueling bandit problem, which is a variant of the $K$-armed bandit problem,  where the feedback comes in the form of pairwise preferences. 
Today's deployed search systems can evaluate a large number of rankers concurrently, and scaling effectively in the presence of numerous rankers is a critical aspect of  $K$-armed dueling bandit problems.

In this paper, we focus on solving the large-scale online ranker evaluation problem under the so-called Condorcet assumption, where there exists an optimal ranker that is preferred to all other rankers. 
We propose \ac{MergeDTS}, which first utilizes a divide-and-conquer strategy that localizes the comparisons carried out by the algorithm to small batches of rankers, and then employs  \ac{TS} to reduce the comparisons between suboptimal rankers inside these small batches.  
The effectiveness (regret) and efficiency (time complexity) of \ac{MergeDTS} are extensively evaluated using examples from the domain of online evaluation for web search. 
Our main finding is that for large-scale Condorcet ranker evaluation problems, \ac{MergeDTS} outperforms the state-of-the-art dueling bandit algorithms.	
\end{abstract}

%
% The code below should be generated by the tool at
% http://dl.acm.org/ccs.cfm
% Please copy and paste the code instead of the example below. 
%



\maketitle
\acresetall

%!TEX root = ../main.tex

\section{Introduction}
\label{sec:intro}

Online ranker evaluation concerns the task of determining the ranker with the best performance out of a finite set of rankers.	
It is an important challenge for information retrieval systems~\citep{markov-2018-what,Ibrahim:2016:CPL:2954381.2866571,Moffat:2017:IUE:3026478.3052768}. 
In the absence of an oracle judge who can tell the preferences between all rankers, the best ranker is usually inferred from user feedback on the result lists produced by the rankers~\cite{hofmann-online-2016}. 
Since user feedback is known to be noisy \cite{joachims03:evaluating,Joachims:2007:EAI:1229179.1229181,Nelissen:2018:STU:3211967.3185153,Goldberg:2018:IDS:3211967.3186195}, how to infer ranker quality and when to stop evaluating a ranker are two important challenges in online ranker evaluation.

The former challenge, i.e., how to infer the quality of a ranker, is normally addressed by \emph{interleaving} methods~\cite{chapelle2012large,hofmann2011probabilistic,interleave,hofmann-pi-2012,Chuklin:2015:CAI:2737813.2668120}. 
Specifically, an interleaving method interleaves the result lists generated by two rankers for a given query and presents the interleaved list to the user.  
Then it infers the preferred ranker based on the user's click feedback. 
As click feedback is noisy, the interleaved comparison of two rankers has to be repeated many times so as to arrive at a reliable outcome of the comparison.

Although interleaving methods address the first challenge of online ranker evaluation (how to infer the quality of a ranker), they give rise to another challenge, i.e., which rankers to compare and when to stop the comparisons.  
Without enough comparisons, we may mistakingly infer the wrong  ranker preferences. But with too many comparisons we may degrade the user experience as we continue to show results from sub-optimal rankers. 
Based on previous work~\cite{,rcs,mergerucb,mdb}, the challenge of choosing and comparing rankers can be formalized as a $K$-armed dueling bandit problem~\cite{kdb}, which is an important variant of the \ac{MAB} problem, where feedback is given in the form of pairwise preferences. 
In the $K$-armed dueling bandit problem, a ranker is defined as an arm and the best ranker is the arm that has the highest expectation to win the interleaving game against other candidates. 

A number of dueling bandit algorithms have been proposed; cf.~\cite{masrour,DBsurvey,DBsurvey2} for an overview. 
However, the study of these algorithms has mostly been limited to small-scale dueling bandit problems, with the state-of-the-art being \ac{DTS}~\cite{dts}. By ``small-scale'' we mean that the number of arms being compared is small. 
But, in real-world online ranker evaluation problems, experiments involving hundreds or even thousands of rankers are commonplace~\cite{kohavi2013online}.
Despite this fact, to the best of our knowledge, the only work that address this particular scalability issue is \ac{MergeRUCB}~\cite{mergerucb}. 
As we demonstrate in this paper, the performance of \ac{MergeRUCB} can be improved upon substantially.

In this paper, we propose and evaluate a novel algorithm, named \ac{MergeDTS}. 
The main idea of \ac{MergeDTS} is to combine the benefits of \ac{MergeRUCB}, which is the state-of-the-art algorithm for large-scale dueling bandit problems, and the benefits of \ac{DTS}, which is the state-of-the-art algorithm for small-scale problems, and attain improvements in terms of effectiveness (as measured in terms of regret) and efficiency (as measured in terms of time complexity). 
More specifically, what we borrow from \ac{MergeRUCB} is the divide and conquer idea used to group rankers into small batches to avoid global comparisons. 
On the other hand, from \ac{DTS} we import the idea of using  Thompson Sampling \cite{tst}, rather than using uniform randomness as in \ac{MergeRUCB}, to choose the arms to be played. 

We analyze the performance of \ac{MergeDTS}, and demonstrate that the soundness of  \ac{MergeDTS} can be guaranteed if the time step $T$ is known and the exploration parameter $\alpha > 0.5$ (Theorem~\ref{th:bound}).
Finally, we conduct extensive experiments to evaluate the performance of \ac{MergeDTS} in the scenario of online ranker evaluation on three widely used real-world datasets: 
Microsoft, Yahoo!~Learning to Rank, and ISTELLA~\cite{mslr,chapelle2011yahoo,Deveaud:2018:LAR:3289475.3231937,istella}. 
We show that with tuned parameters \ac{MergeDTS} outperforms \ac{MergeRUCB} and \ac{DTS} in large-scale online ranker evaluation under the Condorcet assumption, i.e., where there is a ranker preferred to all other rankers.\footnote{
	Our theoretical analysis is rather conservative and the regret bound only holds for the parameter values within a certain range. This is because our bound is proven using Chernoff-Hoeffding bound~\cite{Hoeffding:1963} together with the union bound~\cite{casella2002statistical}, both of which, in our case, introduce gaps between theory and practice. 
	In our experiments, we show that the parameter values outside of the theoretical regime can boost up the performance of \ac{MergeDTS} as well as that of the baselines.
	Thus, our experimental results of \ac{MergeDTS} are not restricted to the parameter values within the theoretical regime.
}
Moreover, we demonstrate the potential of using \ac{MergeDTS} beyond the Condorcet assumption, i.e., where there might be multiple best rankers.

In summary, the main contributions of this paper are as follows:
\begin{enumerate}
	\item 
		We propose a novel $K$-armed dueling bandits algorithm for large-scale online ranker evaluation, called \ac{MergeDTS}. 
		We use the idea of divide and conquer together with Thompson sampling to reduce the number of comparisons of arms. 
	\item 
		We analyze the performance of \ac{MergeDTS} and theoretically demonstrate that the soundness of \ac{MergeDTS} can be guaranteed in the case of known time horizon and parameter values in the theoretical regime. 
	\item 
		We evaluate \ac{MergeDTS} experimentally on the Microsoft, Yahoo! Learning to Rank and ISTELLA datasets, and show that, with the tuned parameters,
		\ac{MergeDTS} outperforms baselines in most of the large-scale online ranker evaluation configurations. 
\end{enumerate}

\noindent
The rest of the paper is organized as follows. 
In Section~\ref{sec:problem-setting}, we detail the definition of the dueling bandit problem. 
We discuss prior work in Section~\ref{sec:relatedwork}. 
MergeDTS is proposed in Section~\ref{sec:mergedts}. 
Our experimental setup is detailed in Section~\ref{sec:exp_setup} and the results are presented in Section~\ref{section:results}.
We conclude in Section~\ref{sec:conclusion}. 

%!TEX root = MergeDTS.tex
\section{Problem Setting}
\label{sec:problem-setting}

In this section, we first describe in more precise terms the $K$-armed dueling bandit problem, which is a variation of the \acf{MAB} problem. 
The latter can be described as follows: 
given $K$ choices, called ``arms'' and denoted by $a_1,\ldots,a_K$, we are required to choose one arm at each step;
choosing arm $a_i$ generates a reward which is drawn i.i.d.\ from a random variable with mean, denoted by $\mu_i$, and our goal is to maximize the expected total reward accumulated by our choices of arms over time.
This objective is more commonly formulated in terms of the \emph{cumulative regret} of the \ac{MAB} algorithm, where regret at step $t$ is the difference between the reward of the chosen arm, e.g., $a_j$, and the reward of the best arm, e.g., $a_k$, in hindsight, and the average regret of arm $a_j$ is defined to be 
$\mu_k-\mu_j$:
cumulative regret is defined to be the sum of the instantaneous regret over time~\cite{ucb,li-2019-bubblerank}. 

The dueling bandit problem differs from the above setting in that at each step we can  choose up to two arms, $a_i$ and $a_j$ ($a_i$ and $a_j$ can be the same); the feedback is either $a_i$ or $a_j$, as the winner of the comparison between the two arms (rather than an absolute reward), where $a_i$ is chosen as the winner with \emph{preference probability} $p_{ij}$ and $a_j$ with probability $p_{ji}=1-p_{ij}$. 
These probabilities form the entries of a $K\times K$ \emph{preference matrix} $\mathbf{P}$, which defines the dueling bandit problem but is not revealed to the dueling bandit algorithm.

In a similar fashion to the \ac{MAB} setting, we evaluate a dueling bandit algorithm based on its \textit{cumulative regret}, which is the total regret incurred by choosing suboptimal arms comparing to the best arm over time~\cite{DBsurvey,masrour}.
However, the definition of regret is less clear-cut in the dueling bandit setting, due to the fact that our dueling bandit problem might not contain a clear winner that is preferred to all other arms, i.e., an arm $a_C$, called the \emph{Condorcet winner}, such that $p_{Cj} > 0.5$ for all $j \neq C$. 
There are numerous proposals in the literature for alternative notions of winners in the absence of a Condorcet winner, e.g., Borda winner \cite{savage,sdb}, Copeland winner \cite{ccb,cwrmed}, von Neumann winner \cite{cdb}, with each definition having its own disadvantages as well as practical settings where its use is appropriate.

MergeDTS, like most of the other dueling bandits algorithms~\cite{DBsurvey2, rcs,rucb,mergerucb,savage}, relies on the existence of a  Condorcet winner, 
in which case the Condorcet winner is the clear choice for the best arm, since it is preferred to all other arms, and with respect to which regret can be defined.
We pose, as an interesting direction for future work, the task of extending the method proposed in this paper to each of the other notions of winner listed above.

In order to simplify the notation in the rest of the paper, we re-label the arms such that $a_1$ is the Condorcet winner, although this is not revealed to the algorithm.
We define the \emph{regret} incurred by comparing $a_i$ and $a_j$ at time $t$ to be  
\begin{equation}
r_t = (\Delta_{1i} + \Delta_{1j})/2 ,
\end{equation} 
where $\Delta_{1k} := p_{1k} - 0.5$ for each $k$. 
Moreover, the \emph{cumulative regret} after $T$ steps is defined to be 
\begin{equation}
\mathcal{R}(T) = \sum_{t=1}^{T} r_t, 
\end{equation}
where $r_t$ is the regret incurred by our choice of arms at time $t$.

Let us translate the online ranker evaluation problem into  the dueling bandit problem.
The input, a finite set of arms, consists of a set of rankers, e.g., based on different ranking models or based on the same model but with different parameters~\cite{kohavi2013online}. 
The Condorcet winner is the ranker that is preferred, by the majority of users, over suboptimal rankers.  
More specifically, a result list from the Condorcet winner is expected to receive the highest number of clicks from users when compared to a list from a suboptimal ranker.  
The preference matrix $\mathbf{P}$ records the users'  relative preferences for all rankers. 
Regret measures the user frustration incurred by showing the interleaved list from suboptimal rankers instead of the Condorcet winner.
In the rest of the paper,  we use the  term  \emph{ranker} to indicate the term \emph{arm} in $K$-armed dueling bandit problems since we focus on the online ranker evaluation task. 
%!TEX root = ../main.tex

\section{Related Work}
\label{sec:relatedwork}

There are two main existing approaches for solving  dueling bandit problems: 
\begin{inparaenum}
\item reducing the problem to a \ac{MAB} problem, e.g., Sparring~\citep{sparring}, Self-Sparring~\cite{selfsparring} and \ac{REX3}~\cite{gajane2015relative};
\item generalizing existing MAB algorithms to the dueling bandit setting, e.g., \ac{RUCB}~\citep{rucb}, \ac{RMED1}~\cite{rmed} and \ac{DTS}~\cite{dts}.
\end{inparaenum}
The advantage of the latter group of algorithms is that they come equipped with theoretical guarantees, proven for a broad class of problems.
The first group, however, have guarantees that either only hold for a restricted class of problems, where the dueling bandit problem is obtained by comparing the arms of an underlying MAB problem (a.k.a. utility-based dueling bandits), as in the case of Self-Sparring,  \ac{REX3} and Sparring T-INF~\cite{zimmert19a-optimal}, or have substantially suboptimal instance-dependent regret bounds as in the case of Sparring EXP3, which has a regret bound of the form $O(\sqrt{KT})$, as opposed to $O(K\log T)$.

% \todo{@Masour: 
% The reviewer pointed out that the outcome of Dudik et al. (2015) can also be applied to general Condorcet case. 
% I didn't read your paper, will you please take a look at this and add some responses?
% Here is the review: 
% \emph{
% 	First, to set the terminology right, the “dueling bandit problem obtained from comparing the arms of an underlying MAB problem” is typically called utility-based dueling bandits.
% 	Dudik et al. (2015) provide an adversarial regret bound against the von Neumann winner for sparring of EXP4 algorithms, which directly implies a regret bound for sparring of EXP3. The adversarial regret bound implies a stochastic regret bound (not the optimal one, but it is a theoretical guarantee). The bound against the von Neumann winner also implies a bound against the Condorcet winner, when one exists. So this approach applies to any dueling bandits, not just utility-based.xt}
% }
Indeed, as our experimental results below demonstrate, Sparring-type algorithms can perform poorly when the dueling bandit problem does not arise from a \ac{MAB} problem. 

Below, we describe some of these algorithms to provide context for our work. 
Sparring~\citep{sparring} uses two \ac{MAB} algorithms, e.g., \ac{UCB}, to choose rankers.  
At each step, Sparring asks each \ac{MAB} algorithm to output a ranker to be compared. 
The two rankers are then compared and the \ac{MAB} algorithm that proposed the winning ranker gets a reward of $1$ and the other a reward of $0$.

Self-Sparring~\citep{selfsparring} improves upon Sparring by employing a single \ac{MAB} algorithm, but at each step samples twice to choose rankers.
More precisely, \citet{selfsparring} use \acf{TS} as the MAB algorithm.  
Self-Sparring assumes that the problem it solves arises from an MAB; it can perform poorly when there exists a cycle relation in rankers, i.e., if there are rankers $a_i$, $a_j $ and $a_k$ with $p_{ij} > 0.5$, $p_{jk}>0.5$ and $p_{ki} > 0.5$. 
As Self-Sparring does not estimate confidence intervals of the comparison results, it does not eliminate rankers.

Another extension of Sparring is \ac{REX3}~\cite{gajane2015relative}, which is designed for the adversarial setting.
\ac{REX3} is inspired by the \ac{EXP3}~\cite{Auer:exp3}, an algorithm for adversarial bandits, and has a regret bound of the form $O( \sqrt{K \ln{(K)}T)}$. 
Note that the regret bound grows as the square-root of time-steps, but sublinearly in the number of rankers, which shows the potential for improvement in the case of large-scale problems. 

\acf{RUCB}~\citep{rucb} extends \ac{UCB} to dueling bandits using a matrix of optimistic estimates of the relative preference probabilities. 
At~each step, \ac{RUCB} chooses the first ranker to be one that beats all other rankers based on the idea of \emph{optimism in the face of uncertainty}. 
Then it chooses the second ranker to be the ranker that beats the first ranker with the same idea of optimism in the face of uncertainty, which translates to pessimism for the first ranker. 
The cumulative regret of \ac{RUCB} after $T$ steps is upper bounded by an expression of the form $O(K^2 + K \log T)$. 

\ac{RMED1}~\cite{rmed} extends an asymptotically optimal \ac{MAB} algorithm, called \ac{DMED} \cite{honda2010asymptotically}, by first proving an asymptotic lower bound on the cumulative regret of all dueling bandit algorithms, which has the order of $\Omega(K \log T)$, and pulling each pair of rankers the minimum number of times prescribed by the lower bound.
\ac{RMED1} outperforms \ac{RUCB} and Sparring. 

\acf{DTS}~\citep{dts} improves upon \ac{RUCB} by using \ac{TS} to break ties when choosing the first ranker. 
Specifically, it uses one \ac{TS} to choose the first ranker from a set of candidates that are pre-chosen by UCB. 
Then it uses another \ac{TS} to choose the second ranker that performs the best compared to the first one. 
The cumulative regret of \ac{DTS} is upper bounded by $O(K \log T +  K^2 \log \log T)$.
{Note that the bound of \ac{DTS} is higher than that of \ac{RUCB}. We hypothesize that this is because the bound of  \ac{DTS} is rather loose.} 
\ac{DTS} outperforms other dueling bandits algorithms empirically and is the state-of-the-art in the case of small-scale dueling bandit problems~\citep{dts,selfsparring}. 
As discussed in Section~\ref{sec:exp_setup}, for computational reasons \ac{DTS} is not suitable for large-scale problems. 

The work that is the closest to ours is by~\citet{mergerucb}.
They propose \ac{MergeRUCB}, which is the state-of-the-art for large-scale dueling bandit problems.
\ac{MergeRUCB} partitions rankers into small batches and compares rankers within each batch. 
A ranker is eliminated from a batch once we realize that even according to the most optimistic estimate of the preference probabilities it loses to another ranker in the batch. 
Once enough rankers have been eliminated, \ac{MergeRUCB} repartitions the remaining rankers and continues as before. 
Importantly, \ac{MergeRUCB} does not require global pairwise comparisons between all pairs of rankers, and so it reduces the computational complexity and increases the time efficiency, as shown in Section~\ref{sec:computational}.
The cumulative regret of \ac{MergeRUCB} can be upper bounded by $O(K \log T)$~\cite{mergerucb}, i.e., with no quadratic dependence on the number of rankers.  
This upper bound has the same order as the lower bound proposed by \citet{rmed} in terms of $K \log{T}$, but it is not optimal in the sense that it has large constant coefficients. 
As we demonstrate in our experiments,  
 \ac{MergeRUCB} can be improved by making use of \ac{TS} to reduce the amount of randomness in the choice of rankers. 
More precisely, the cumulative regret of \ac{MergeRUCB} is almost twice as large as that of \ac{MergeDTS} in the large-scale setup shown in Section~\ref{section:results}. 

A recent extension of dueling bandits is called \emph{multi-dueling bandits}~\cite{mdb,selfsparring,battleduel}, where more than two rankers can be compared at each step. 
\ac{MDB} is the first proposed algorithm in this setting, which is specifically designed for online ranker evaluation. 
It maintains two UCB estimators for each pair of rankers, a looser confidence bound and a tighter one. 
At each step, if there is more than one ranker that is valid for the tighter UCB estimators, \ac{MDB} compares all the rankers that are valid for the looser UCB estimators.   
\ac{MDB} is outperformed by Self-Sparring, the state-of-the-art \emph{multi-dueling bandit} algorithm,  significantly~\cite{selfsparring}. 
In this paper, we do not focus on the \emph{multi-dueling bandit} setup.
The reasons are two-fold.
First, to the best of our knowledge, there are no theoretical results in the multi-dueling setting that allow for the presence of cyclical preference relationships among the rankers. 
Second, \citet{battleduel} state that ``(perhaps surprisingly) [\ldots] the flexibility of playing size-$k$ subsets does not really help to gather information faster than the corresponding dueling case ($k=2$),  at least for the current subset-wise feedback choice model.''
This statement demonstrates that there is no clear advantage to using multi-dueling comparisons over pairwise dueling comparisons at this moments.

%!TEX root = ../main.tex

\algrenewcommand{\algorithmiccomment}[1]{// #1}
\algnewcommand\algorithmicinput{\textbf{Input:}}
\algnewcommand\Input{\item[\algorithmicinput]}
\algnewcommand\algorithmicoutput{\textbf{Output:}}
\algnewcommand\Output{\item[\algorithmicoutput]}

\section{Merge Double Thompson Sampling}
\label{sec:mergedts}
In this section, we describe the proposed algorithm, \ac{MergeDTS}, and explain the main intuition behind it.
Then, we provide theoretical guarantees bounding the regret of \ac{MergeDTS}. 

\subsection{The \ac{MergeDTS} algorithm}

Here we describe MergeDTS, Merge Double Thompson Sampling, which combines the benefits of both the elimination-based divide and conquer strategy of MergeRUCB and the sampling strategy of DTS, producing an effective scalable dueling bandit algorithm.

The pseudo-code for \ac{MergeDTS} is provided in Algorithms~\ref{alg:mergedts}--\ref{alg:phase2}, 
with the notation summarized in Table~\ref{tb:notation} for the reader's convenience. 
The input parameters are the exploration parameter $\alpha$, the size of a batch $M$ and the  failure probability $\epsilon \in (0, 1)$. 
The algorithm records the outcomes of the past comparisons in matrix $\mathbf{W}$, whose element $w_{ij}$ is the number of times ranker $a_i$ has beaten ranker $a_j$ so far.
\ac{MergeDTS} stops when only one ranker remains and then returns that ranker, which it claims to be the Condorcet winner.\footnote{In the online ranker evaluation application, we can stop \ac{MergeDTS} once it finds the best ranker. However, in our experiments, we keep \ac{MergeDTS} running by comparing the remaining ranker with itself. If the remaining ranker is the Condorcet winner, there will be no regret. }

\begin{table}
	\caption{Notation used in this paper.}
	\label{tb:notation}
	\begin{tabular}{p{0.15\columnwidth} p{0.75\columnwidth}}
		\toprule	
		\bf Notation & \bf Description \\
		\midrule
		$K$ & Number of rankers \\
		$a_i$ & The $i$-th ranker\\
		$p_{ij}$ & Probability of $a_i$ beating $a_j$ \\
		$M$ & Size of a batch \\ 
		$\alpha$ & Exploration parameter, $\alpha > 0.5$ \\
		$\epsilon$ & Probability of failure \\
		$\mathbf{W}$ & The comparison matrix \\
		$w_{ij}$ & Number of times $a_i$ has beaten $a_j$ \\
		$s$ & Stage of the algorithm \\
		$\mathcal{B}_s$ & Set of batches at the $s$-th stage \\
		$b_s$ & Number of batches in $\mathcal{B}_s$ \\
		$\theta_{ij}$  & Sampled probability of $a_i$ beating $a_j$ \\
		$a_c$ & Ranker chosen in Phase~I of MergeDTS \\
		$\phi_{i}$ & Sampled probability of $a_i$ beating $a_c$ \\
		$a_d$ & Ranker chosen in Phase~II of MergeDTS \\
		$u_{ij}$ & Upper confidence bound (UCB): $\frac{w_{ij}}{w_{ij} + w_{ji}} + \sqrt{\frac{\alpha \log{(t+C(\epsilon))}}{w_{ij}+w_{ji}}}$ \\
		% \st{$l_{ij}$} & \st{Lower confidence bound: $1-u_{ji}$ } \cl{We don't really use this in the paper.} \\
		$\Delta_{ij}$ & $|p_{ij} - 0.5|$ \\
		$\Delta_{\min}$ & $\min_{\Delta_{ij} > 0} \Delta_{ij}$ \\
		$\Delta_{B, min}$ & $ \min_{a_i, a_j \in B ~and ~i\neq j}\Delta_{ij}$\\
		% $\mathit{kl}(p\| q)$ & $p \ln\frac{p}{q} + (1-p) \ln\frac{1-p}{1-q}$ \\ 
		% KL-divergence of  $\Ber(p)$ and $\Ber(q)$\\
		% $\mathit{kl}_{\min}$ & $\min_{i\neq j}\mathit{kl}(p_{ij}\|0.5)$\\
		$C(\epsilon)$ & $ \left(\frac{(4\alpha -1)K^2}{(2\alpha-1)\epsilon}\right) ^{\frac{1}{2\alpha-1}} $  \\
		% $\Delta_{\min}$ & $\min\limits_{\{ (i, j)\mid p_{ij} \neq 1/2  \} }\Delta_{i,j}^2 $ \\
		\bottomrule
	\end{tabular}
\end{table}

\begin{algorithm}[!t]
	\begin{algorithmic}[1]
		\Input $K$ rankers $a_1, a_2, \ldots, a_K$; partition size $M$; exploration parameter $\alpha>0.5$; running time steps $T$;  probability of failure $\epsilon = 1/T$.
		\Output The Condorcet winner. 
		\State $\mathbf{W} \leftarrow \mathbf{0}_{K,K}$ \hfill \Comment \emph{The comparison matrix}
		\State  \label{alg:C} $C(\epsilon) = \left(\frac{(4\alpha-1)K^2}{(2\alpha-1)\epsilon}\right)^{\frac{1}{2\alpha-1}}$
		\State $s=1$ \hfill \Comment \emph{The stage of the algorithm}
		\State $\mathcal{B}_s = \big\{\underbrace{[a_1, \ldots, a_M]}_{B_1}, \ldots, \underbrace{[a_{(b_1-1)M+1}, \ldots, a_K]}_{B_{b_1}}\big\}$ \label{alg:group_batches}
		\hfill \Comment \emph{Disjoint batches of rankers, with $b_1 = \lceil\frac{K}{M}\rceil $}\nonumber
		
		\For{$t = 1,2,\ldots T $} 
			\State $m = t \mod b_s$ \hfill \Comment \emph{Index of the batches}
			
			\If{$b_s = 1$ and $|B_m|=1$}\label{alg:oneranker} \hfill \Comment One ranker left
			\State\label{alg:bestranker} Return the remaining ranker $a\in B_m$.
			\EndIf			
			\State $\mathbf{U} = \frac{\mathbf{W}}{\mathbf{W} + \mathbf{W}^{T}}+ \sqrt{\left(\frac{\alpha \log(t + C(\epsilon))}{\mathbf{W} + \mathbf{W}^{T}}\right)}$ \label{alg:ubound}
			\hfill \Comment \emph{UCB estimators: operations are element-wise and $\frac{x}{0} := 1$}\nonumber
			
			\State Remove $a_i$ from $B_m$ if $u_{ij} < 0.5$ for any $a_j \in B_m$.
			\label{alg:remove}
			\If{$b_s > 1$ and $|B_m|=1$} 
				\State Merge $B_m$ with the next batch and decrement $b_s$.%
				\label{alg:merge}
			\EndIf
			
			\hspace{-2mm}\Comment \textbf{Phase~I}: Choose the first candidate $a_c$\nonumber
			\State $a_c = $ SampleTournament($\mathbf{W}$, $B_m$) \hfill \Comment \emph{See Algorithm~\ref{alg:phase1}}\label{alg:first_tournament}

			\hspace{-2mm}\Comment \textbf{Phase~II}: Choose the second candidate $a_d$
			\State $a_d$ = RelativeTournament($\mathbf{W}$, $B_m$, $a_c$) \hfill \Comment \emph{See Algorithm~\ref{alg:phase2}}
			
			\hspace{-2mm}\Comment \textbf{Phase~III}: Compare candidates and update batches
			\State Compare pair $(a_c, a_d)$ and increment $w_{cd}$ if $a_c$ wins otherwise increment $w_{dc} $.
			\label{alg:update}

			\hspace{-2mm}\Comment \textbf{Phase~IV}: Update batch set
			\If{$\sum_m|B_m| \leq \frac{K}{2^s}$} \label{alg:update_batch}
			\State Pair the larger size batches with the smaller ones, making sure the size of every batch is in $[0.5M, 1.5M]$.
			\State $s = s+1$
			\State Update $\mathcal{B}_s$, $b_s = |\mathcal{B}_s|$.
			\EndIf\label{alg:update_batch_end} 
		\EndFor
	\end{algorithmic}
	\caption{MergeDTS (Merge Double Thompson Sampling)}
	\label{alg:mergedts}
\end{algorithm}

\ac{MergeDTS} begins by grouping rankers into small batches (Line~\ref{alg:group_batches}). 
At each time-step, \ac{MergeDTS} checks whether there is more than one ranker remaining (Line~\ref{alg:oneranker}). 
If so, \ac{MergeDTS} returns that ranker, the potential Condorcet winner. 
If not, \ac{MergeDTS} considers one batch $B_m$ and, using optimistic estimates of the preference probabilities (Line~\ref{alg:ubound}), it purges any ranker that loses to another ranker even with an optimistic boost in favor of the former (Line~\ref{alg:remove}).

If, as a result of the above purge, $B_m$ becomes a single-element batch, it is merged with the next batch $B_{m+1}$ (Line~\ref{alg:merge}).
Here, $m+1$ is interpreted as modulo $b_s$, where $b_s$ is the number of batches in the current stage.
This is done to avoid comparing a suboptimal ranker against itself, since if there is more than one batch, the best  ranker in any given batch is unlikely to be the Condorcet winner of the whole dueling bandit problem.
As we will see again below, MergeDTS takes great care to avoid comparing suboptimal rankers against themselves because it results in added regret, but yields no extra information, since we know that each ranker is tied with itself.

After the above elimination step, the algorithm proceeds in four phases: 
choosing the first ranker (Phase~I), choosing the second ranker based on the first ranker (Phase II), comparing the two rankers and updating the statistics (Phase III), and repartitioning the rankers at the end of each stage (Phase~IV).
Of the four phases, Phase~I and Phase~II are  the major reasons that lead to a boost in effectiveness of \ac{MergeDTS} when compared to \ac{MergeRUCB}. 
We will elaborate both phases in the remainder of this section. 

In Phase~I, the method \emph{SampleTournament} (Algorithm~\ref{alg:phase1}) chooses the first candidate ranker: \ac{MergeDTS} samples preference probabilities $\theta_{ij}$ from the posterior distributions to estimate the true preference probabilities $p_{ij}$ for all pairs of rankers in the batch $B_m$ (Lines~\ref{alg:first_sample}--\ref{alg:first_sample_end}, the first \ac{TS}).
Based on these sampled probabilities, \ac{MergeDTS} chooses the first candidate $a_c$ so that it beats most of the other rankers according to the sampled preferences (Line~\ref{alg:first_arm}). 

\begin{algorithm}[!h]
	\begin{algorithmic}[1]
		\Input  The comparison matrix $\mathbf{W}$ and the current batch $B_m$.
		\Output The first candidate  $a_c$.
		\For{$a_i, a_j \in B_m$ and $i < j$}	\label{alg:first_sample}
		\State Sample $\theta_{ij} \sim Beta(w_{ij}+1, w_{ji}+1)$ 
		\State $\theta_{ji} = 1 - \theta_{ij}$	
		\EndFor\label{alg:first_sample_end}
		
		\State $\kappa_i = \frac{1}{|B_m|-1} \sum_{a_j \in B_m, j\neq i} \mathds{1}(\theta_{ij}>0.5)$
		\State $a_c = \underset{a_i \in B_m}{\argmax} ~ \kappa_i$; \emph{breaking ties randomly} \hfill \Comment \emph{First candidate} \label{alg:first_arm}
		
	\end{algorithmic}
	\caption{SampleTournament }
	\label{alg:phase1}
\end{algorithm}

In Phase~II, the method \emph{RelativeTournament} (Algorithm~\ref{alg:phase2}) chooses the second candidate ranker: \ac{MergeDTS} samples another set of preference probabilities $\phi_j$ from the posteriors of $p_{jc}$ for all rankers $a_j$ in $B_m \setminus \{a_c\}$ (Lines~\ref{alg:second_sample}--\ref{alg:second_sample_end}, the second \ac{TS}).
Moreover, we set $\phi_c$ to be~$1$ (Line~\ref{alg:ac_nonfinal}). 
This is done to avoid self-comparisons between suboptimal rankers for the reasons that were described above.

\begin{algorithm}[!h]
	\begin{algorithmic}[1]
		\Input  The comparison matrix $\mathbf{W}$, the current batch $B_m$ and the first candidate $a_c$.
		\Output The second candidate  $a_d$.
			\For{$a_j \in B_m$ and  $j \neq c$}\label{alg:second_sample}
				\State Sample $\phi_j \sim Beta(w_{jc} +1, w_{cj} +1)$
			\EndFor\label{alg:second_sample_end}
			\State $\phi_c = 1$ \hfill \Comment \emph{Avoid self-comparison}
			\label{alg:ac_nonfinal}
			\State $a_d = \underset{a_j \in B_m}{\argmin} ~ \phi_j$; \emph{breaking ties randomly}\hfill \Comment \emph{Second candidate}
			\label{alg:second_arm_worst}
	\end{algorithmic}
	\caption{RelativeTournament}
	\label{alg:phase2}
\end{algorithm}

\noindent%
Once the probabilities $\phi_j$ have been sampled, we choose the ranker $a_d$ that is going to be compared against $a_c$, using the following strategy. 
The worst ranker according to the sampled probabilities $\phi_j$ is chosen as the second candidate $a_d$ (Line~\ref{alg:second_arm_worst}).
The rationale for this discrepancy is that we would like to eliminate rankers as quickly as possible, so rather than using the upper confidence bounds to explore when choosing $a_d$, we use the lower confidence bounds to knock the weakest link out of the batch as quickly as possible.

In Phase~III (Line~\ref{alg:update}) of Algorithm~\ref{alg:mergedts}, \ac{MergeDTS} plays $a_c$ and $a_d$ and updates the comparison matrix $\mathbf{W}$ based on the observed feedback. 

Finally, in Phase~IV (Lines~\ref{alg:update_batch}--\ref{alg:update_batch_end}), if the number of remaining rankers in the current stage is half of the rankers of the previous stage (Line~\ref{alg:update_batch}), \ac{MergeDTS} enters the next stage, before which it repartitions the rankers. 
Following the design of \ac{MergeRUCB}, this is done by merging batches of rankers such that the smaller sized batches are combined with the larger sized batches; 
we enforce that the number of rankers in the new batches is kept in the range of $[0.5M, 1.5M]$.

\subsection{Theoretical guarantees}

% There are three types of comparison in \ac{MergeDTS}: comparisons between one suboptimal arm and others, self-comparisons of suboptimal arms, and self-comparisons of the Condorcet winner. 
% The third type does not increase the regret. 
% We only need to bound the numbers of the first and second type of comparison to bound the cumulative regret of \ac{MergeDTS}. 

In this section, we state and prove a high probability upper bound on the regret accumulated by MergeDTS after $T$ steps, under the assumption that the dueling bandit problem contains a Condorcet winner. 
Since the theoretical analysis of MergeDTS is based on that of \ac{MergeRUCB}, we start by listing two assumptions that we borrow from \ac{MergeRUCB} in~\cite[Section~7]{mergerucb}.
\begin{description}[align=left, leftmargin=*]
 	\item[Assumption 1.] There is no repetition in rankers. All ranker pairs $(a_i, a_j)$ with $i \neq j$ are distinguishable, i.e., $p_{ij} \neq 0.5$, unless both of them are ``uninformative'' rankers that provide random ranked lists and cannot beat any other rankers. 
 	\item[Assumption 2.] The uninformative rankers are at most one third of the full set of rankers.  
\end{description}

\noindent%
These assumptions arise from the Yahoo! Learning to Rank Challenge dataset, where there are $181$ out of $700$ rankers that always provide random ranked lists. 
\textbf{Assumption~1} ensures that each informative ranker is distinguishable from other rankers. 
\textbf{Assumption~2} restricts the maximal  percentage of uninformative rankers and thus ensures that the probability of triggering the merge condition (Line~\ref{alg:update_batch} in Algorithm~\ref{alg:mergedts}) is larger than $0$.\footnote{
In practice, \ac{MergeDTS} works without \textbf{Assumption~2} because the Condorcet winner eliminates all other arms eventually with $O(K^2\log{T})$ comparisons. 
We keep \textbf{Assumption~2} to ensure that \ac{MergeDTS} also works in cases where we have the $O(K\log(T))$ guarantee. 
We refer readers to \cite{mergerucb} for a detailed discussion.
}
Moreover, we emphasize that \textbf{Assumption~1} and \textbf{Assumption~2} are milder than the assumptions made in Self-Sparring and DTS, where indistinguishability is simply not allowed.

We now state our main theoretical result:

\begin{theorem}
	\label{th:bound}
	With the known time step $T$, applying \ac{MergeDTS} with $\alpha > 0.5$, $M \geq 4$ and $\epsilon = 1/T$ to a $K$-armed Condorcet dueling bandit problem under \textbf{Assumption~1} and \textbf{Assumption~2}, with probability $1-\epsilon$ the cumulative regret $\mathcal{R}(T)$ after $T$ steps is bounded by: 
	\begin{equation}\label{eq:bound}
	\mathcal{R}(T) < \frac{8\alpha M K \ln(T + C(\epsilon))}{\Delta_{\min}^2}, 
	\end{equation}
where
\begin{equation}\label{eq:Delta_min}
\Delta_{\min} := \min_{\Delta_{ij} > 0} \Delta_{ij},
\end{equation}
is the minimal gap of two distinguishable rankers and $C(\epsilon) =  \left(\frac{(4\alpha -1)K^2}{(2\alpha-1)\epsilon}\right) ^{\frac{1}{2\alpha-1}} $.
\end{theorem}

\noindent% 
The upper bound on the $T$-step cumulative regret of \ac{MergeDTS} is $O(K\ln{(T)} /\Delta_{\min}^2)$. 
In other words, the cumulative regret grows linearly with the number of rankers, $K$.   
This is the most important advantage of \ac{MergeDTS}, which states the potential of applying it to the large-scale online evaluation. 
We emphasize that for most of the $K$-armed dueling bandit algorithms in the literature, the upper bounds contain a $K^2$ term, which renders them unsuitable for large-scale online ranker evaluation.
By the definition of $\Delta_{\min}$ in Equation \eqref{eq:Delta_min} we have $\Delta_{\min} > 0$, and so our bound is well-defined. 
However, the performance of \ac{MergeDTS} may degrade severely when $\Delta_{\min}$ is small. 
$\alpha$ is a common parameter in UCB-type algorithms, called the exploration parameter. 
$\alpha$ controls the trade-off between exploitation and exploration: larger $\alpha$ results in more exploration, whereas smaller $\alpha$ makes the algorithm more exploitative. 
Theoretically, $\alpha$ should be larger than $0.5$.
However, as shown in our experiments, using some values of $\alpha$ that are outside the theoretical regime can lead to a boost in the effectiveness of \ac{MergeDTS}. 

% \mz{
% Let us however point out a weakness of the bound provided above: if applied with a fixed $\epsilon$ to an infinite horizon setting, MergeDTS may have linear regret, as is the case with 
% }
% \ch{
% 	% With the infinite time step, $t$, MergeDTS may have the linear regret. 
% 	This is the nature of elimination-based bandit algorithm.
% 	There is alway a probability $\epsilon$ that the algorithm eliminates the best arm. 
% 	For the infinite case, where $t >> 1/\epsilon$,  the algorithm has the linear regret on average. 
% 	However, as discussed above, in this paper, we focus on the finite $T$ case, where the regret of MergeDTS is sublinear on average. 
% }

Theorem~\ref{th:bound} provides a finite-horizon high probability bound. 
From a practical point of view, this type of bound is of great utility. 
In practice, bandit algorithms are always deployed and evaluated within limited user iterations~\cite{dbgd,li-contextual-2010}. 
Here, each time step is one user interaction. 
As the number of interactions is provided, we can choose a reasonable step $T$ to make sure the high probability bound holds. 
We can also get an expected regret bound of \ac{MergeDTS} at step $T$ by setting $\epsilon = 1/T$ and adding $1$ to the right-hand side of \eqref{eq:bound}: 
this is because $\mathbb{E}[\mathcal{R}(T)]$ can be bounded by
\begin{equation}
 \frac{1}{T} \cdot T + \frac{T-1}{T} \cdot \frac{8\alpha M K \ln(T + C(\epsilon))}{\Delta_{\min}^2} \leq 1 + \frac{8\alpha M K \ln(T + C(\epsilon))}{\Delta_{\min}^2}. 
\end{equation}
We note that the above expected bound holds only at time-step $T$ and so the horizonless version of \ac{MergeDTS} does not possess an expected regret bound.

The proof of Theorem~\ref{th:bound} relies on the Lemma~3 in~\cite{mergerucb}. We repeat it here for the reader's convenience. 
\begin{lemma}[Lemma~3 in \cite{mergerucb}]
\label{lm:mergerucb}
	Given any pair of distinguishable rankers $a_i, a_j \in B$ and $\epsilon \in [0, 1]$, with the probability of $1-\epsilon$, the maximum number of comparisons that could have been carried out between these two rankers in the first $T$ time-steps before a merger between B and another batch occurs, is bounded by 
	\begin{equation}
			\frac{4\alpha \ln (T + C(\epsilon))}{\Delta_{B, min}^2}, 
	\end{equation}
	where $\Delta_{B, min} = \min_{a_i, a_j \in B ~and ~i\neq j}\Delta_{ij}$ is the minimal gap of two distinguishable rankers in batch $B$.  
\end{lemma}

\begin{proof}[Proof of Theorem~\ref{th:bound}]
Lemma~\ref{lm:mergerucb} states that with probability $1-\epsilon$ the number of comparisons between a pair of distinguishable rankers $(i, j) \in B$ is bounded by 
\begin{equation}
\frac{4\alpha \ln (T + C(\epsilon))}{\Delta_{B, min}^2},
\end{equation}
regardless of the way the rankers are selected, as long as the same criterion as \ac{MergeRUCB} is used for eliminating rankers.
Since the elimination criterion for MergeDTS is the same as that of MergeRUCB, we can apply the same argument used to prove Theorem~1 in \cite{mergerucb} to get a bound of 
\begin{equation}
\frac{8\alpha M K \ln(T+ C(\epsilon)) }{\Delta_{\min}^2}
\end{equation}
on the regret accumulated by MergeDTS. 
Here we use the fact that $\Delta_{B, min} \geq \Delta_{min}$ and thus $\frac{4\alpha \ln (T + C(\epsilon))}{\Delta_{B, min}^2} \leq \frac{4\alpha \ln (T + C(\epsilon))}{\Delta_{min}^2}$. 
\qedhere
\end{proof}

\subsection{Discussion}

The prefix ``merge'' in \ac{MergeDTS} signifies the fact that it uses a similar divide-and-conquer strategy as merge sort. 
It partitions the $K$-arm set into small batches of size $M$. 
The comparisons only happen between rankers in the same batch, which, in turn, avoids global  pairwise comparisons and gets rid of the $O(K^2)$ dependence in the cumulative regret, which is the main limitation for using dueling bandits for large datasets. 

In contrast to sorting, \ac{MergeDTS} needs a large number of comparisons before declaring a difference between rankers since the feedback is stochastic. 
The harder two rankers are to distinguish or in other words the closer $p_{ij}$ is to $0.5$, the more comparisons are required. 
Moreover, if a batch only contains the uninformative rankers, the comparisons between those rankers will not stop, which incurs infinite regret. 
\ac{MergeDTS} reduces the number of comparisons between hardly distinguishable rankers as follows:
\begin{enumerate}
	\item \ac{MergeDTS} compares the best ranker in the batch to the worst to avoid comparisons between hardly distinguishable rankers; 
	\item when half of the rankers of the previous stage are eliminated, \ac{MergeDTS} pairs larger batches to smaller ones that contain at least one informative ranker and enters the next stage. 
\end{enumerate}
The second item is borrowed from the design of \ac{MergeRUCB}. 

\ac{MergeDTS} and  \ac{MergeRUCB} follow the same ``merge'' strategy.
The difference between these two algorithms is in their strategy of choosing rankers, i.e., Algorithms~\ref{alg:phase1} and~\ref{alg:phase2}. 
\ac{MergeDTS} employs a sampling strategy to choose the first ranker inside the batch and then uses another 
sampling strategy to choose the second ranker that is potentially beaten by the first one. 
As stated above, this design comes from the fact that \ac{MergeDTS} is  carefully designed to reduce the comparisons between barely distinguishable rankers. 
In contrast to \ac{MergeDTS}, \ac{MergeRUCB} randomly chooses the first ranker and chooses the second ranker to be the one that is the most likely to beat the first ranker, as discussed in Section~\ref{sec:relatedwork}.
The uniformly random strategy inevitably increases the number of comparisons between those barely distinguishable rankers. 

In summary, the double sampling strategy used by \ac{MergeDTS} is the major factor that leads to the superior performance of \ac{MergeDTS} as demonstrated by our experiments. 

%!TEX root = ../main.tex

\section{Experimental Setup}
%\label{sec:setup}
\label{sec:exp_setup}

\subsection{Research questions}
In this paper, we investigate the application of dueling bandits to the large-scale online ranker evaluation setup. 
Our experiments are designed to answer the following research questions:
\begin{enumerate}[align=left]
	\item[\textbf{RQ1}] Does \ac{MergeDTS} outperform the state-of-the-art large-scale algorithm \ac{MergeRUCB} as well as the more recently proposed Self-Sparring in terms of cumulative regret, i.e., effectiveness?
\end{enumerate}
In the bandit literature~\cite{DBsurvey2,DBsurveyold,DBsurvey}, regret is a measure of the rate of convergence to the Condorcet winner in hindsight. 
Mapping this to the online ranker evaluation setting, \textbf{RQ1} asks whether  \ac{MergeDTS} hurts the user experience less than baselines while it is being used for large-scale online ranker evaluation.
\begin{enumerate}[align=left]
	\item[\textbf{RQ2}]  How do \ac{MergeDTS} and the baselines scale computationally?
\end{enumerate}
What is the time complexity of \ac{MergeDTS}? 
Does \ac{MergeDTS} require less running time than the baselines? 

\begin{enumerate}[align=left]
	\item[\textbf{RQ3}] How do different levels of noise in the feedback signal affect cumulative regret of \ac{MergeDTS} and the baselines?
\end{enumerate}
In particular, can we still observe the same results in \textbf{RQ1} after a (simulated) user changes its behavior? 
How sensitive are \ac{MergeDTS} and the baselines to noise? 

\begin{enumerate}[align=left]
	\item[\textbf{RQ4}] How do \ac{MergeDTS} and the baselines perform when the Condorcet dueling bandit problem contains cycles?
\end{enumerate}
Previous work has found that cyclical preference relations between rankers are abundant in online ranker comparisons~\cite{rcs,mergerucb}. 
Can \ac{MergeDTS} and the baselines find the Condorcet winner when the experimental setup features a large number of cyclical relations between rankers?

\begin{enumerate}[align=left]
	\item[\textbf{RQ5}] How does \ac{MergeDTS} perform when the dueling bandit problem violates the Condorcet assumption?
\end{enumerate}
We focus on the Condorcet dueling bandit task in this paper. 
Can \ac{MergeDTS} be applied to the dueling bandit tasks without the existence of a Condorcet winner? 
\begin{enumerate}
	\item[\textbf{RQ6}] Which approach finds the best ranker faster: \ac{MergeDTS} together with \ac{PI} or \ac{SOSM}? 
\end{enumerate}
There are two general approaches to evaluate rankers online: \begin{inparaenum}
	\item  using dueling bandits together with interleaving~\cite{mergerucb,rcs};
	\item directly using multileaving methods~\cite{Brost:2016}. 
\end{inparaenum}
The former approach can output the best ranker with high confidence,
while the later approach can compare multiple rankers simultaneously which may shorten the comparison process.

\begin{enumerate}
	\item[\textbf{RQ7}] What is the parameter sensitivity of \ac{MergeDTS}?
\end{enumerate}
Can we improve the performance of \ac{MergeDTS} by tuning its parameters, such as the exploration parameter $\alpha$, the size of a batch $M$, and the probability of failure~$\epsilon$?

\subsection{Datasets}
\label{sec:dataset}
To answer our research questions, we use two types of dataset: three real-world datasets and a synthetic dataset.

	First, to answer \textbf{RQ1}--\textbf{RQ3} and \textbf{RQ6}, we run experiments on three large-scale datasets: the Microsoft Learning to Rank (MSLR) WEB30K dataset~\cite{mslr}, the Yahoo!\ Learning to Rank Challenge Set 1 (Yahoo)~\cite{chapelle2011yahoo} and the ISTELLA dataset~\cite{istella}.\footnote{
	We omit the Yahoo Set 2 dataset because it contains far fewer queries than the Yahoo Set 1 dataset.}
These datasets contain a large number of features based on unsupervised ranking functions, such as BM25, TF-IDF, etc.
In our experiments, we take the widely used setup in which each individual feature is regarded as a ranker~\cite{mergerucb,hofmann2011probabilistic}.  
This is different from a real-world setup, where a search system normally ranks documents using a well trained learning to rank algorithm that combines multiple features.
However, the difficulty of a dueling bandit problem comes from the relative quality of pairs of rankers and not from their absolute quality.  
In other words, evaluating rankers with similar and possibly low performance is as hard as evaluating state-of-the-art rankers, e.g., LambdaMART~\cite{lambdamart}. 
Therefore, we stick to the standard setup of~\cite{mergerucb,hofmann2011probabilistic}, treating each feature as a ranker
and each ranker as an arm in the $K$-armed dueling bandit problem. 
We leave experiments aimed at comparing different well trained learning to rank algorithms as future work. 
As a summary, the MSLR dataset contains $136$ rankers, the Yahoo dataset contains $700$ rankers and the ISTELLA dataset contains $220$ rankers.
Compared to the typical $K$-armed dueling bandit setups, where $K$ is generally substantially smaller than $100$~\cite{selfsparring,btm,sparring,dts}, these are large numbers of rankers. 

Second, to answer \textbf{RQ4}, we use a synthetic dataset, generated by~\citet{mergerucb}, which contains cycles (called the \emph{Cycle dataset} in the rest of the paper). 
The Cycle dataset has $20$ rankers with one Condorcet winner, $a_1$, and $19$ suboptimal rankers, $a_2, \ldots, a_{20}$.  
The Condorcet winner beats the other $19$ suboptimal rankers. 
And those $19$ rankers have a cyclical preference relationship between them. 
More precisely, following~\citet{mergerucb}, the estimated probability $p_{1j}$ of $a_1$ beating $a_j$  ($j=2, \ldots, 20$ is set to $p_{1j} = 0.51$, and the preference relationships between the suboptimal rankers are described as follows: visualize the $19$ rankers $a_2, \ldots, a_{20}$ sitting at a round table, then each ranker beats every ranker to its left with probability $1$ and loses to every ranker to its right with probability $1$.
In this way we obtain the Cycle dataset.

Note that this is a difficult setup for Self-Sparring, because Self-Sparring chooses rankers based on their Borda scores,
and the Borda scores ($\sum_{j=1}^{K}p_{ij}$ for each ranker $a_i$~\cite{savage}) are close to each other in the Cycle dataset.
For example, the Borda score of the Condorcet winner is $10.19$, while the Borda score of a suboptimal ranker is $9.99$.
This makes it hard for Self-Sparring to distinguish between rankers.
In order to be able to conduct a fair comparison, we generate the \emph{Cycle2 dataset}, where each suboptimal ranker beats every ranker to its left with a probability of $0.51$ and the Condorcet winner beats all others with probability $0.6$. 
Now, in the Cycle2 dataset, the Borda score of the Condorcet winner is $11.90$, while the Borda score of a suboptimal ranker is $9.9$. 
Thus, it is an easier setup for Self-Sparring. 
 
Furthermore, to answer \textbf{RQ5}, we use the MSLR-non-Condorcet dataset from~\cite{dts}, which is a subset of the MSLR dataset that does not contain a Condorcet winner. 
This datasets has $32$ rankers with two Copeland winners (instead of one), each of which beats the other $30$ rankers. 
A Copeland winner is a ranker that beats the largest number of other rankers; every dueling bandit dataset contains at least one Copeland winner~\cite{ccb}.

Finally, we use the MSLR dataset with the navigational configuration (described in Section~\ref{sec:clicksimulation}) to assess the parameter sensitivity of \ac{MergeDTS}~(\textbf{RQ7}).

\subsection{Evaluation methodology}
To evaluate dueling bandit algorithms, we follow the proxy approach from~\cite{mergerucb}.
It first uses an interleaving algorithm to obtain a preference matrix, i.e., a matrix that for each pair of rankers contains the probability that one ranker beats the other.
More precisely, for each pair of rankers $a_i$ and $a_j$, $p_{ij}$ is the estimation that $a_i$ beats $a_j$ in the simulated interleaved comparisons. 
Then, this obtained preference matrix is used to evaluate dueling bandit algorithms: 
for two rankers $a_i$ and $a_j$ chosen by a dueling bandit algorithm, we compare them by drawing a sample from a Bernoulli distribution with mean $p_{ij}$, i.e., $1$ means that $a_i$ beats $a_j$ and vice versa. 
This is a standard approach to evaluating dueling bandit algorithms~\cite{btm,dts,rcs}. 
Moreover, the proxy approach has been shown to have the same quality as interleaving in terms of evaluating dueling bandit algorithms~\cite{mergerucb}.

In this paper, we adopt the procedure described by \citet{mergerucb}, who use Probabilistic Interleave~\citep{hofmann2011probabilistic} to obtain a preference matrix for the MSLR dataset, and obtain a preference matrix for the Yahoo datasets.\footnote{We use the implementation of Probabilistic Interleave in the LEROT software package~\cite{lerot}.}
In the case of the MSLR dataset, the number of comparisons for every pair of rankers is \numprint{400000} and, in the case of the Yahoo dataset, we use \numprint{60000} comparisons per pair of rankers.
The reason for this discrepancy is pragmatic: the latter dataset has roughly $27$ times as many pairs of rankers to be compared.

\subsection{Click simulation}
\label{sec:clicksimulation}
Since the interleaved comparisons mentioned above are carried out using click feedback, we follow~\citet{interleave} and simulate clicks using three configurations of a click model~\cite{cm}: namely \emph{perfect}, \emph{navigational} and \emph{informational}.
The perfect configuration simulates a user who checks every document and clicks on a document with a probability proportional to the query-document relevance. 
This configuration is the easiest one for dueling bandit algorithms to find the best ranker, because it contains very little noise. 
The navigational configuration mimics a user who seeks specific information, i.e., who may be searching for the link of a website, and is likely to stop browsing results after finding a relevant document. 
The navigational configuration contains more noise than the perfect configuration and is harder for dueling bandit algorithms to find the best ranker. 
Finally, the informational configuration represents a user who wants to gather all available information for a query and may click on documents that are not relevant with high probability. 
In the informational configuration the feedback contains more noise than in the perfect and navigational configurations, which makes it the most difficult configuration for dueling bandit algorithms to determine the best ranker, which, in turn, may result in the highest cumulative regret among the three configurations.

To answer the research questions that concern large-scale dueling bandit problems, namely \textbf{RQ1}, \textbf{RQ2}, \textbf{RQ6} and \textbf{RQ7}, we use the navigational configuration, which represents a reasonable middle ground between the perfect and informational configurations~\cite{hofmann2011probabilistic}.
The corresponding experimental setups are called MSLR-Navigational, Yahoo-Navigational and ISTELLA-Navigational.
To answer \textbf{RQ3} regarding the effect of feedback with different levels of noise,
we use all three configurations on MSLR, Yahoo and ISTELLA datasets: namely MSLR-Perfect, MSLR-Navigational and MSLR-Informational; Yahoo-Perfect, Yahoo-Navigational and Yahoo-Informational; ISTELLA-Perfect, ISTELLA-Navigational and ISTELLA-Informational. 
Thus, we have nine large-scale setups in total.

\subsection{Baselines}
We compare \ac{MergeDTS} to five state-of-the-art dueling bandit algorithms:
\ac{MergeRUCB}~\cite{mergerucb}, \ac{DTS}~\cite{dts}, \ac{RMED1}~\cite{rmed}, Self-Sparring~\cite{selfsparring}, and REX3~\cite{gajane2015relative}. 
Among these algorithms, \ac{MergeRUCB} is designed for large-scale online ranker evaluation and is the state-of-the-art large-scale dueling bandit algorithm. 
\ac{DTS} is the state-of-the-art small-scale dueling bandit algorithm.
\ac{RMED1} is motivated by the lower bound of the Condorcet dueling bandit problem and matches the lower bound up to a factor of $O(K^2)$, which indicates that \ac{RMED1} has low regret in small-scale problems but may have large regret when the number of rankers is large. 
Self-Sparring is a more recently proposed dueling bandit algorithm that is the state-of-the-art algorithm in the multi-dueling setup, with which multiple rankers can be compared in each step.
REX3 is proposed for the adversarial dueling bandit problem but also performs well for the large-scale stochastic dueling bandit problem~\cite{gajane2015relative}. 
We do not include \ac{RUCB}~\cite{rucb} and Sparring~\cite{sparring} in our experiments since they have been outperformed by more than one of our baselines~\cite{mergerucb,selfsparring,gajane2015relative}. 

\subsection{Parameters}
\label{sec:parameters setup}

Recall that Theorem~\ref{th:bound} is based on Lemma~3 in~\cite{mergerucb}.
The latter provides a high probability guarantee that the confidence intervals will not mislead the algorithm into eliminating the Condorcet winner by mistake.
However, this result is proven using the Chernoff-Hoeffding \cite{Hoeffding:1963} bound together with an application of the union bound \cite{casella2002statistical}, both of which introduce certain gaps between theory and practice.
That is, the analysis of regret mainly considers the worst-case scenario rather than the average-case scenario, which makes regret bounds much looser than they could have been.
We conjecture that the expression for $C(\epsilon)$, which derives its form from Lemma~3 in~\cite{mergerucb}, is excessively conservative.
Put differently, Theorem~\ref{th:bound} specifies a sufficient condition for the proper functioning of \ac{MergeDTS}, not a necessary one. 
So, a natural question that arises is the following: to what extent can restrictions imposed by our theoretical results be violated without the algorithm failing in practice?
In short, what is the gap between theory and practice and what is the parameter sensitivity of \ac{MergeDTS}? 

To address these questions and answer \textbf{RQ7}, we conduct extensive parameter sensitivity analyses in the MSLR-Navigational setup with the following parameters: 
$\alpha \in \{0.8 ^{0}$, $0.8^{1}$, \dots, $0.8^{9}\}$,
$C \in \{4 \times 10^2$, $4\times 10^3, \ldots, 4 \times 10^6, \numprint{4726908}\}$, and $M \in \{2$, $4$,  $8$, $16\}$. 
$C$ is short for $C(\epsilon)$,	where $C(\epsilon)$ is the exploration bonus added to the confidence intervals. 
According to Table~\ref{tb:notation}, $C(\epsilon)$ is a function of $\alpha$ and $\epsilon$. 
However, to simplify our experimental setup, we consider $C$ as an individual parameter rather than a function parameterized by $\alpha$ and $\epsilon$, and study the impact of $C$ directly. 
The details about  the choice of the values are explained in the following paragraph.

When choosing candidate values for $\alpha$, we want them to cover the optimal theoretical value $\alpha > 1$, the lowest theoretically legal value $\alpha > 0.5$, and for smaller values of $\alpha$ we want to decrease the differences between two consecutive $\alpha$'s. 
This last condition is imposed because smaller values of $\alpha$ may mislead \ac{MergeDTS} to eliminate the Condorcet winner. 
So we shrink the search space for smaller values of $\alpha$.
The powers of $0.8$ from $0$ to $9$ seem to satisfy the above conditions, particularly $0.8^3 \approx 0.5$ and obviously $0.8^0 = 1$ with the difference between $0.8^n$ and $0.8^{n+1}$ becoming smaller with larger $n$.
The value $C = \numprint{4726908}$ is calculated from the definition of $C(\epsilon)$ with the default $\alpha=1.01$ and $M=4$ (see Table~\ref{tb:notation}), noting that the MSLR-Navigational setup contains $136$  rankers, i.e., $K=136$.
As discussed before, the design of $C(\epsilon)$ may be too conservative. 
So, we only choose candidate values smaller than $\numprint{4726908}$.
We use the log-scale of $C(\epsilon)$ because the upper bound is logarithmic with $C(\epsilon)$.

The sensitivity of parameters is analyzed by following the order of their importance to Theorem~\ref{th:bound}, i.e., $\alpha$ and $M$ have a linear relation to the cumulative regret and $C$ has a logarithmic relation to the cumulative regret.
We first evaluate the sensitivity of $\alpha$ with the default values of $M$ and $C$.
Then we use the best value of $\alpha$ to test a range of values of $M$ (with default $C$).
Finally, we analyze the impact of $C$ using the best values of $\alpha$ and $M$. 

We discover the practically optimal parameters for \ac{MergeDTS} to be $\alpha=0.8^6$, $M=16$ and $C=\numprint{4000000}$, in Section~\ref{sec:parameter}. 
We repeat the procedure for \ac{MergeRUCB} and \ac{DTS}, and use their optimal parameter values in our experiments,
which are $\alpha = 0.8^6$, $M=8$, $C=\numprint{400000}$ for \ac{MergeRUCB} and $\alpha = 0.8^7$ for \ac{DTS}.
Then, we use these values to answer \textbf{RQ1}--\textbf{RQ5}. 
Self-Sparring does not have any parameters, so further analysis and tuning are not needed here.

	The shrewd readers may notice that the parameters are somewhat overtuned in the MSLR-Navigational setup, and \ac{MergeDTS} with the tuned parameters does not enjoy the theoretical guarantees in Theorem~\ref{th:bound}. 
	However, because of the existence of the gap between theory and practice, we want to answer the question whether we can improve the performance of MergeDTS as well as that of the baselines by tuning the parameters outside of their theoretical limits. 
	We also want to emphasize that the parameters of MergeDTS and baselines are only tuned in the MSLR-Navigational setup, but MergeDTS is compared to baselines in nine setups. 
	If, with the tuned parameters, MergeDTS outperforms baselines in other eight setups, we can also show the potential of improving MergeDTS in an empirical way. 

\subsection{Metrics}
In our experiments, we assess the efficiency (time complexity) and effectiveness (cumulative regret) of \ac{MergeDTS} and baselines. 
The metric for efficiency is the running time in days. 
We compute the running time from the start of the first step to the end of the $T$-th step, where $T = 10^8$ in our experiments.  
A commercial search engine may serve more than $1$ billion search requests per day~\cite{url:google}, and each search request can be used to conduct one dueling bandit comparison. 
The total number of time steps, i.e. $T$, considered in our experiments is about $1\%$ of the one-week traffic of a commercial search engine.

We use cumulative regret in $T$ steps to measure the effectiveness of algorithms, which is computed as follows:
\begin{equation}
    \mathcal{R}(T) = \sum_{t = 1}^{T} r(t) =   \sum_{t = 1}^{T}  \frac{1}{2}(\Delta_{1,c_t} + \Delta_{1,d_t}),
\end{equation}
where $r(t)$ is the regret at step $t$, $c_t$ and $d_t$ are the indices of rankers chosen at step $t$, and without loss of generality, we assume $a_1$ to be the Condorcet winner. 
The regret $r(t)$ arises from the comparisons between the two suboptimal rankers at $t$ step. 
It is the average sub-optimality of comparing two rankers $a_i$ and $a_j$ with respect to the Condorcet winner $a_1$, i.e., $\frac{p_{1i} + p{1j}}{2} - 0.5$. 
In a real-world scenario, we have a fixed time period to conduct our online ranker evaluation, and thus, the number of steps $T$ can be estimated beforehand. 
In our online ranker evaluation task, the $T$-step cumulative regret is related to the drop in user satisfaction during the evaluation process, i.e. higher regret means larger degradation of user satisfaction, because the preference $p_{1i}$ can be interpreted as the probability of the Condorcet winner being preferred to ranker $i$.

Unless stated differently, the results in all experiments are averaged over $50$ and $100$ independent runs on large- and small-scale datasets respectively, where both numbers are either equal to or larger than the choices in previous studies~\cite{dts,selfsparring,mergerucb}. 
In the effectiveness experiments, we also report the standard error of the average cumulative regret, which measures the differences between the average of samples and the expectation of the sample population. 

All experiments on the MSLR and Yahoo datasets are conducted on a server with Intel(R) Xeon(R) CPU E5-2650  $2.00$GHz ($32$ Cores) and $64$ Gigabyte.
All experiments on the ISTELLA dataset are conducted on servers with Intel(R) Xeon(R) Gold 5118 CPU @ $2.30$GHz ($48$ Cores) and $256$ Gigabyte. 
To be precise, an individual run of  each algorithm is conducted on a single core with 1 Gigabyte. 
%!TEX root = ../main.tex

\section{Results}
\label{section:results}

\newcommand{\plotheight}{0.98}
\newcommand{\wideplotheight}{0.68}

In this section, we analyze the results of our experiments.
In Section~\ref{sec:large-scale}, we compare the effectiveness (cumulative regret) of \ac{MergeDTS} and the baselines in three large-scale online evaluation setups. 
In Section~\ref{sec:computational}, we compare and analyze the efficiency (time complexity) of \ac{MergeDTS} and the baselines. 
In Section~\ref{sec:noise}, we study the impact of different levels of noise in the click feedback on the algorithms. 
In Section~\ref{sec:cycle} and Section~\ref{sec:noncondorcet}, we evaluate \ac{MergeDTS} and the baselines in two alternative setups: the cyclic case and the non-Condorcet case, respectively. 
In Section~\ref{sec:multileaving}, we compare \ac{MergeDTS} to multileaving methods. 
In Section~\ref{sec:parameter}, we analyze the parameter sensitivity of \ac{MergeDTS}.

\subsection{Large-scale experiments}
\label{sec:large-scale}

\begin{figure}[t]
	\centering
	\includegraphics{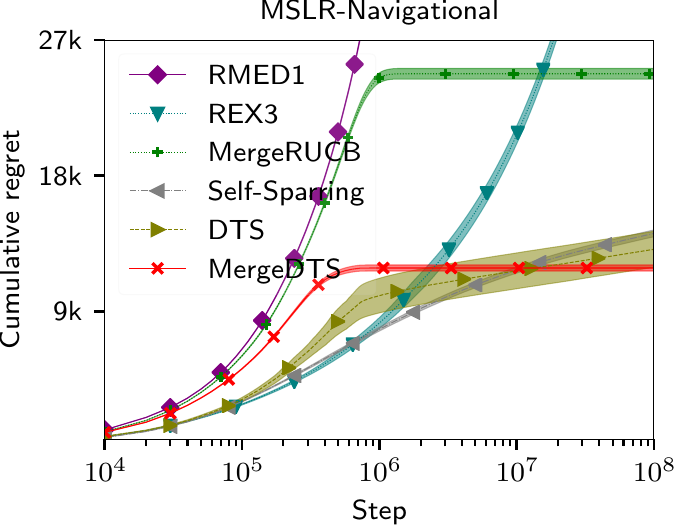}
	\includegraphics{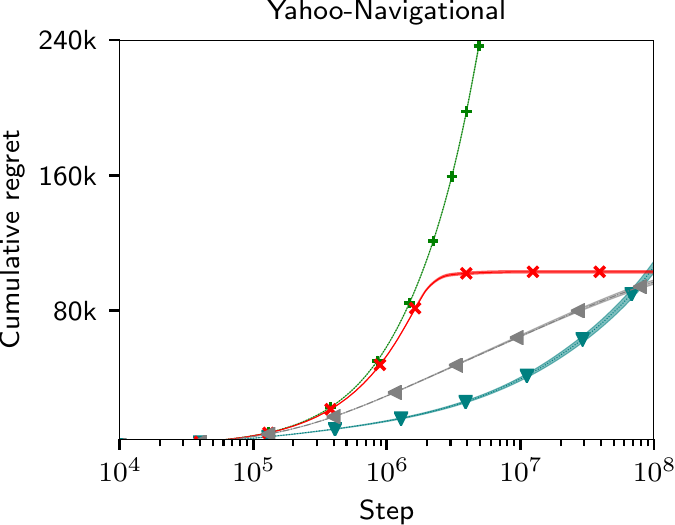}
	\includegraphics{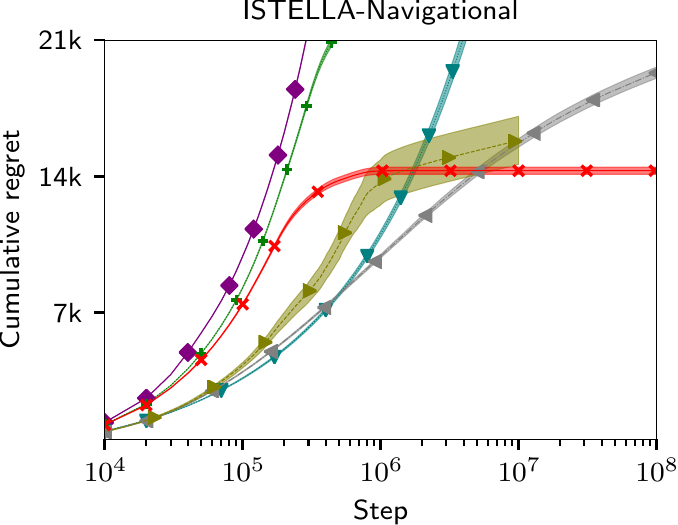}	
	\vspace*{1mm}
	\caption{Cumulative regret on large-scale online ranker evaluation: lower is better. 
	Note that the scales of the y-axes are different. 
	The shaded areas are $\pm$ standard error. 
	The results are averaged over $50$ independent runs.
	}  
	\label{fig:large}	
\end{figure}

To answer \textbf{RQ1}, we compare \ac{MergeDTS} to the large-scale state-of-the-art baseline, \ac{MergeRUCB}, as well as the more recently proposed Self-Sparring, in three large-scale online evaluation setups, namely MSLR-Navigational, Yahoo-Navigational and ISTELLA-Navigational. 
The results are reported in Figure~\ref{fig:large}, which depicts the cumulative regret of each algorithm, averaged over $50$ independent runs. 
As mentioned in Section~\ref{sec:exp_setup}, cumulative regret measures the rate of convergence to the Condorcet winner in hindsight and thus lower regret curves are better.
\ac{DTS} and \ac{RMED1} are not considered in the Yahoo-Navigational setup, and the cumulative regret of \ac{DTS} is reported with in $10^7$ steps on the ISTELLA dataset, because of the computational issues,
which are further discussed in Section~\ref{sec:computational}.
Figure~\ref{fig:large} shows that MergeDTS  outperforms the large-scale state-of-the-art MergeRUCB with large gaps. 
Regarding the comparison with Self-Spar\-ring and DTS, we would like to point out the following facts:
\begin{inparaenum}[(1)]
\item Merge\-DTS outperforms \ac{DTS} and Self-Sparring in MSLR-Navigational and ISTELLA-Navigational setups; 
\item 
	In the Yahoo-Navigational setup, \ac{MergeDTS} has slightly higher cumulative regret than Self-Sparring, but \ac{MergeDTS} converges to the Condorcet winner after three million steps while the cumulative regret of Self-Sparring is still growing after $100$ million steps.
\end{inparaenum}
Translating these facts into real-world scenarios, 
\ac{MergeDTS} has higher regret compared to DTS and Self-Sparring in the early steps, but \ac{MergeDTS} eventually outperforms DTS and Self-Sparring with longer experiments.
As for \ac{REX3}, we see that it has a higher order of regret than other algorithms since \ac{REX3} is designed for the adversarial dueling bandits and the regret of \ac{REX3} is $O(\sqrt{T})$. 
In this paper, we consider the stochastic dueling bandits, and the regret of the other algorithms is $\log{(T)}$. 
\ac{MergeDTS} outperforms the baselines in most setups, but we need to emphasize that the performance of \ac{MergeDTS} here cannot be guaranteed by Theorem~\ref{th:bound}.
This is because we use the parameter setup outside of the theoretical regime, as discussed in Section~\ref{sec:parameters setup}. 

\subsection{Computational scalability}
\label{sec:computational}

\begin{table}[h]
	\caption{
	Average running time in days of each algorithm on large-scale problems for $10^8$ steps averaged over $50$ independent runs. 
	The running time of $\ac{DTS}$ in the ISTELLA-Navigational setup is estimated based on the running time with $10^7$ steps multiplied by $10$. 
 	The running time of \ac{DTS} and \ac{RMED1} in the Yahoo-Navigational setup is estimated by multiplying the average running time at $10^5$ steps by $10^3$.
 		The experiments on the ISTELLA dataset are conducted on different computer clusters from those on the MSLR and Yahoo datasets.
 		The speed of the former ones is about one time faster than the latter ones. 
 		Therefore the numbers may not be directly compared. 
	}
	\label{tb:running_time}
	\centering
	\begin{tabular}{lrrr}
		\toprule
		& MSLR & ISTELLA & Yahoo  \\
		\midrule
		\# Rankers & $136$  &\phantom{10}$220$ & \phantom{10}$700$ \\
		\midrule
		\ac{MergeDTS} & $ 0.08$ &\phantom{10} $0.03$ & \phantom{10}$0.11$ \\
		\ac{MergeRUCB} & $0.08$ &\phantom{10} $0.03$& \phantom{10}$0.11$ \\
		Self-Sparring & $0.18$ &\phantom{10} $0.22$ & \phantom{10}$0.90$ \\
		\ac{DTS}& $5.23$ &\emph{9.88} & $\emph{100.27}$ \\
		\ac{RMED1}& $0.36$ &\phantom{10} $ 0.19$ & \phantom{0}$\emph{18.39}$ \\
		\ac{REX3}& $0.25$ &\phantom{10} $ 0.11$  & \phantom{10}$0.27$ \\				
		\bottomrule
	\end{tabular}
\end{table}

\noindent%
To address \textbf{RQ2}, we report in Table~\ref{tb:running_time} the average running time (in days) of each algorithm in three large-scale dueling bandit setups, namely MSLR-Navigational, ISTELLA and Yahoo-Navigational.
As before, each algorithm is run for $10^8$ steps.
An individual run of  DTS  and \ac{RMED1} in the Yahoo-Navigational setup takes around $100.27$ and $18.39$ days, respectively, which is simply impractical for our experiments; therefore, the running time of DTS and \ac{RMED1} in this setup is estimated by multiplying the average running time at $10^5$ steps by $10^3$.
For a similar reason, we estimate the running time of DTS on the ISTELLA-Navigational setup by multiplying the average running time at $10^7$ by $10$. 

Table~\ref{tb:running_time} shows that \ac{MergeDTS} and \ac{MergeRUCB} have very low running times. 
This is due to the fact that they perform computations inside batches
and their computational complexity is $O(TM^2)$, where $T$ is the number of steps and $M$ is the size of batches. 
Moreover, \ac{MergeDTS} is considerably faster in the MSLR-Navigational setup, because there it finds the best ranker with fewer steps, as can be seen in Figure~\ref{fig:large}.
After finding this best ranker, the size of batches $M$ becomes $1$ and, from that moment on, \ac{MergeDTS} does not perform any extra computations.
The running time of Self-Sparring is also low, but grows with the number of rankers, roughly linearly.
This is because at each step Self-Sparring draws a sample from the posterior distribution of each ranker and its running time is $\Omega(TK)$,
where $K$ is the number of rankers.
\ac{DTS} is orders of magnitude slower than other algorithms and its running time grows roughly quadratically, because \ac{DTS} requires a sample for each pair of rankers at each step and its running time is $\Omega(TK^2)$. 
 
Large-scale commercial search systems process over a billion queries a day \cite{url:google}, and run hundreds of different experiments \cite{kohavi2013online} concurrently, in each of which two rankers are compared.  
The running time for \ac{DTS} and \ac{RMED1} that appears in Table~\ref{tb:running_time} is far beyond the realm of what might be considered reasonable to process on even $20\%$ of one day's traffic: note that one query in the search setting corresponds to one step for a dueling bandit algorithm, since each query could be used to compare two rankers by performing an interleaved comparison. 
Given the estimated running times listed in Table~\ref{tb:running_time}, we exclude \ac{DTS} and \ac{RMED1} from our experiments on the large-scale datasets for practical reasons.

\subsection{Impact of noise}
\label{sec:noise}

\begin{figure}[t]
	\centering
	\includegraphics{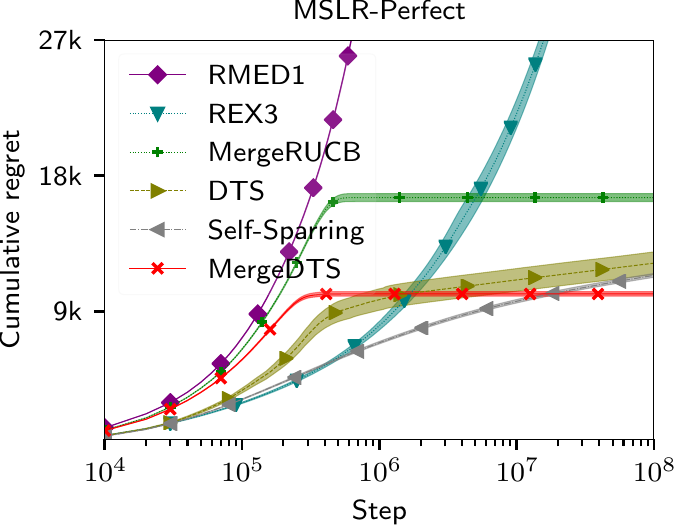}
	\includegraphics{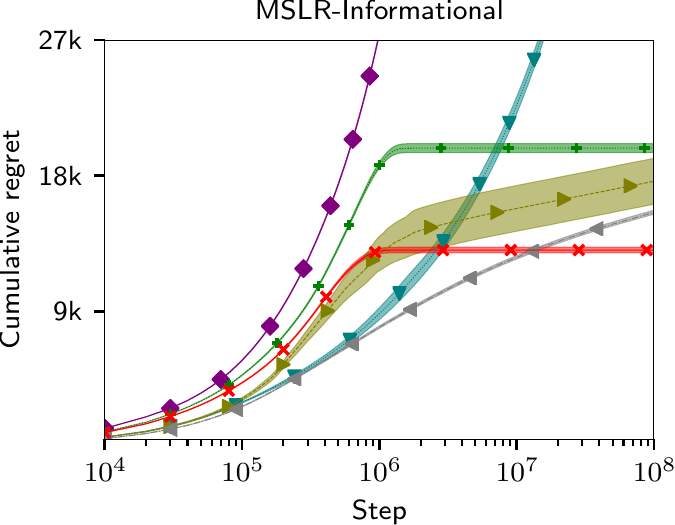}
	\vspace*{1mm}
	\includegraphics{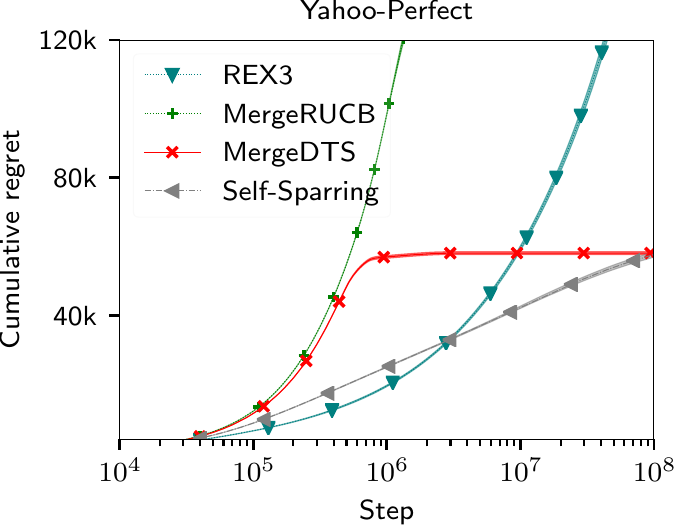}
	\includegraphics{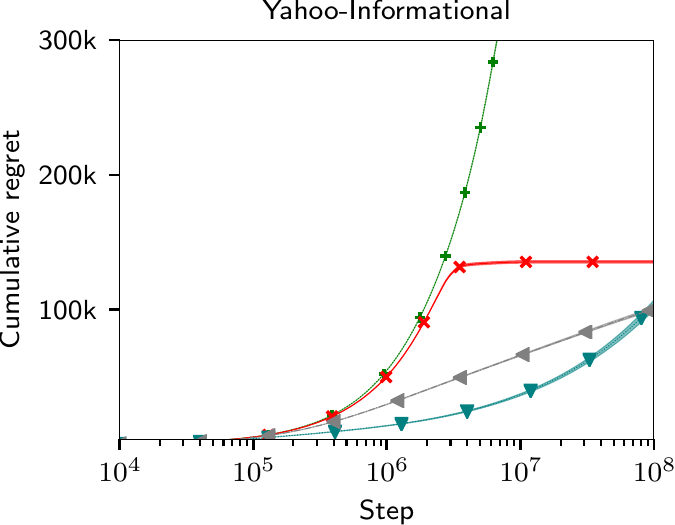}
	\vspace*{1mm}
	\includegraphics{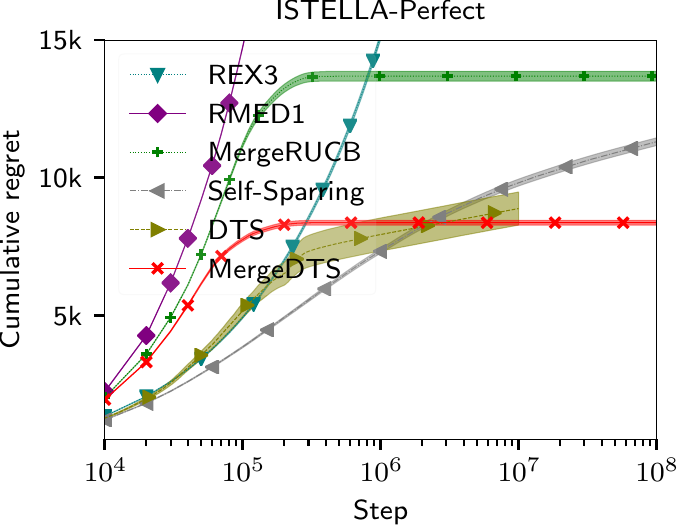}
	\includegraphics{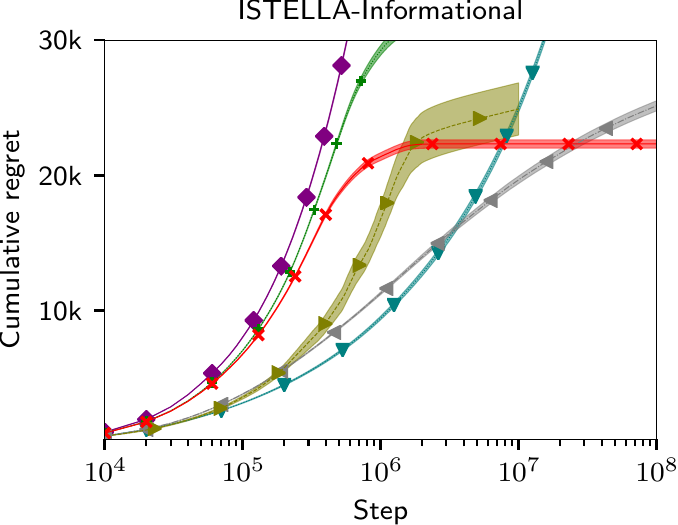}
	\vspace*{1mm}
	\caption{
		Effect  of the level of noise on cumulative regret in click feedback (left column: perfect configuration and right column: informational configuration). 
		The results are averaged over $50$ independent runs. (Contrast with figures in Figure~\ref{fig:large}.)
	}
	\label{fig:clicks}	
\end{figure}
To address \textbf{RQ3}, we run \ac{MergeDTS} in the perfect, navigational and informational configurations (see Section~\ref{sec:clicksimulation}). 
As discussed in Section~\ref{sec:clicksimulation}, the perfect configuration is the easiest one for dueling bandit algorithms to reveal the best ranker, while the informational configuration is the hardest. 
We report the results of the perfect and informational configurations in Figure~\ref{fig:clicks}. 
For comparison, we refer readers to plots in Figure~\ref{fig:large} for the results of navigational configuration. 

On the MSLR and ISTELLA datasets, in all three configurations, \ac{MergeDTS} with the chosen parameters outperforms the baselines, and the gaps get larger as click feedback gets noisier.
The results also show that Self-Sparring is severely affected by the level of noise. 
This is because Self-Sparring estimates the Borda score~\cite{savage} of a ranker and, in our experiments, the noisier click feedback is, the closer the Borda scores are to each other, making it harder for Self-Sparring to identify the winner. 

Results on the Yahoo dataset disagree with results on the MSLR dataset. 
On the Yahoo dataset, \ac{MergeDTS} is affected more severely by the level of noise than Self-Sparring. 
This is because of the existence of uninformative rankers as stated in \textbf{Assumption~1}. 
In noisier configurations, the gaps between uninformative and informative rankers are smaller, which results in the long time of comparisons for \ac{MergeDTS} to eliminate the uninformative rankers. 
Comparing those uninformative rankers leads to high regret. 

In summary, the performance of MergeDTS is largely affected when the gaps between rankers are small, which is consistent with our theoretical findings.

\begin{figure}
	\centering
	\includegraphics{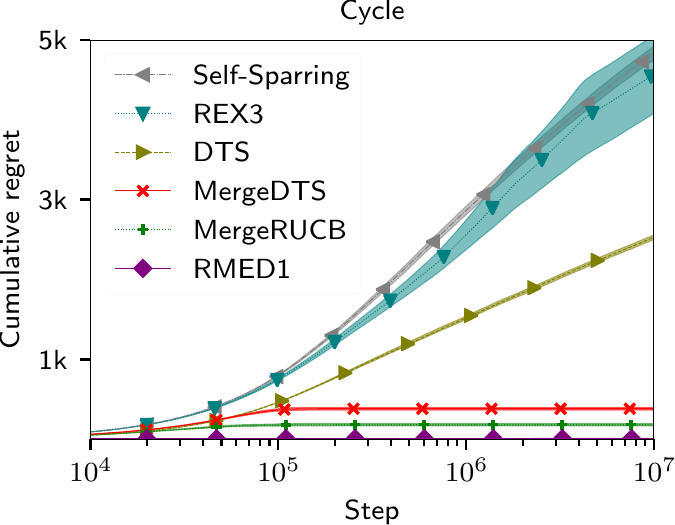}
	\includegraphics{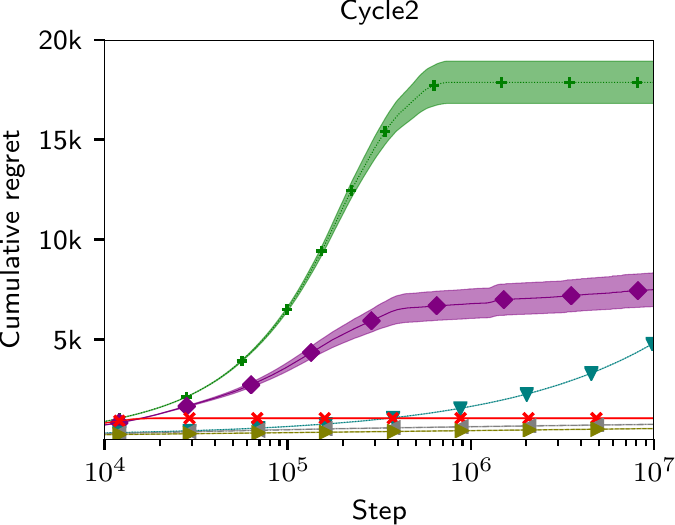}
	\vspace*{1mm}
	\caption{Cumulative regret in the cycle setup. The results are averaged over $100$ independent runs.}
	\label{fig:cycle}
\end{figure}	

\subsection{Cycle experiment}
\label{sec:cycle}

We address \textbf{RQ4} by running the algorithms that we consider on the Cycle and Cycle2 datasets introduced in Section~\ref{sec:dataset}.
Particularly, we have already observed that Self-Sparring performs well in some cases (see the above experiments and results), but we argue that Self-Sparring may perform poorly when a dueling bandit problem contains cyclic preference relationships. 
This has been identified as a point of grave concern in online evaluation~\cite{rcs}.
Therefore, in this section we assess how dueling bandit algorithms behave when a dueling bandit problem contains cycles.

In this section  we conduct experiments for 10 million steps and repeat $100$ times since Merge-style algorithms converge to the Condorcet ranker  within less than 1 million steps and running longer only increases the gaps between \ac{MergeDTS} and baselines. 

For the Cycle dataset (left plot in Figure~\ref{fig:cycle}), 
the cumulative regret of Self-Sparring is an order of magnitude higher than that of \ac{MergeDTS}, 
although it performs well in some cases (see the above experiments). 
As we discussed in Section~\ref{sec:dataset},  Self-Sparring chooses rankers based on their Borda scores  and when 
the Borda scores of different arms become close to each other as in the Cycle dataset, Self-Sparring may perform poorly. 
Also, we  notice that when the gaps in Borda scores of the Condorcet winner and other rankers are large, Self-Sparring performs well, as shown in the right plot in Figure~\ref{fig:cycle}. 

Other than Self-Sparring, we also notice that the other baselines performs quite differently on the two cyclic configurations. 
In the harder configuration of the two, the Cycle dataset, only 
\ac{RMED1} and \ac{MergeRUCB} outperform \ac{MergeDTS}.
\ac{RMED1} excludes rankers from consideration based on relative preferences between two rankers. 
And, in the Cycle dataset, the preferences between suboptimal rankers are large. 
Thus, \ac{RMED1} can easily exclude a ranker based on its relative comparison to another suboptimal ranker. 
For the Cycle2 dataset, where  the relative preferences between two rankers are small, \ac{RMED1} performs worse than \ac{MergeDTS}. 

\ac{MergeRUCB} also slightly outperforms \ac{MergeDTS} on the Cycle dataset. 
This can be explained as follows.
In the Cycle dataset, the preference gap between the Condorcet winner and suboptimal rankers is small (i.e., $0.01$), while the gaps between suboptimal rankers are relatively large (i.e., $1.0$).
Under this setup, \ac{MergeDTS} tends to use the Condorcet winner to eliminate suboptimal rankers in the final stage. 
On the other hand, \ac{MergeRUCB} eliminates a ranker by another ranker who beats it with the largest probability. 
So, \ac{MergeDTS} requires more comparisons to eliminate suboptimal rankers than \ac{MergeRUCB}.
However, the gap between \ac{MergeRUCB} and \ac{MergeDTS} is small.

\begin{figure}
	\centering
	\includegraphics{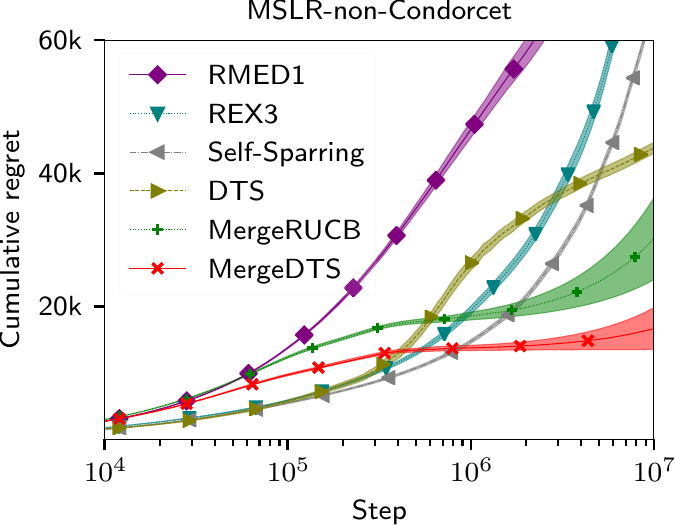}		
	\vspace*{1mm}
	\caption{Cumulative regret in the non-Condorcet setup. The results are averaged over $100$ independent runs.}
	\label{fig:copeland}
\end{figure}

\subsection{Beyond the Condorcet assumption}
\label{sec:noncondorcet}
To answer \textbf{RQ5}, we evaluate \ac{MergeDTS} on the MSLR-non-Condorcet dataset that does not contain a Condorcet winner. 
Instead, the dataset contains two Copeland winners and this dueling bandit setup is called the Copeland dueling bandit~\cite{ccb,dts}.
The Copeland winner is selected by the Copeland score 
$\zeta_i = \frac{1}{K-1} \sum_{k\neq i} \mathds{1}(p_{ik} > 1/2)$ that measures the number of rankers beaten by ranker $a_i$. 
The Copeland winner is defined as $\zeta^* = \max_{1\le i \le K}\zeta_i$. 
In the MSLR-non-Condorcet dataset, each Copeland winner beats $30$ other rankers. 
Specifically, one of the Copeland winners beats the other one but is beaten by a suboptimal ranker. 
In the Copeland dueling bandit setup, regret is computed differently from the Condorcet dueling bandit setup. 
Given a pair of rankers $(a_i, a_j)$, regret at step $t$  is computed as: 
\begin{equation}
r_t =  \zeta^* - 0.5(\zeta_i + \zeta_j).
\end{equation}
Among the considered algorithms, only \ac{DTS} can solve the Copeland dueling bandit problem and is the state-of-the-art Copeland dueling bandit algorithm. 
We conduct the experiment for $10$ million steps with which \ac{DTS} converges to the Copeland winners. 
And we run each algorithm $100$ times independently. 
The results are shown in Figure~\ref{fig:copeland}.

\ac{MergeDTS} has the lowest cumulative regret. 
However, in our experiments, we find that \ac{MergeDTS} eliminates the two Copeland winners one time out of $100$ individual repeats. 
In the other $99$ repeats, we find that \ac{MergeDTS} eliminates one of the two existing winners, which may not be ideal in practice.
Note that we evaluate \ac{MergeDTS} in a relatively easy setup, where only two Copeland winners are considered. 
For more complicated setups, where more than two Copeland winners are considered or the Copeland winners are beaten by several  suboptimal rankers, we speculate that \ac{MergeDTS} can fail more frequently. 
In our experiments, we do not evaluate \ac{MergeDTS} in the more complicated setups, because \ac{MergeDTS} is designed for the Condorcet dueling bandits and is only guaranteed to work under the Condorcet assumption. 
The answer to \textbf{RQ5} is that \ac{MergeDTS} may perform well for some easy setups that go beyond the Condorcet assumption without any guarantees. 

\subsection{Comparison to Multileaving}
\label{sec:multileaving}

\begin{figure*}
	\centering
	\includegraphics {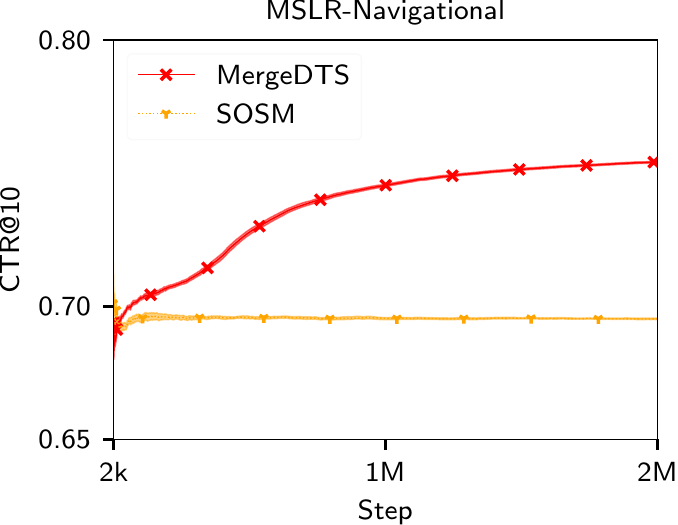}
	\includegraphics {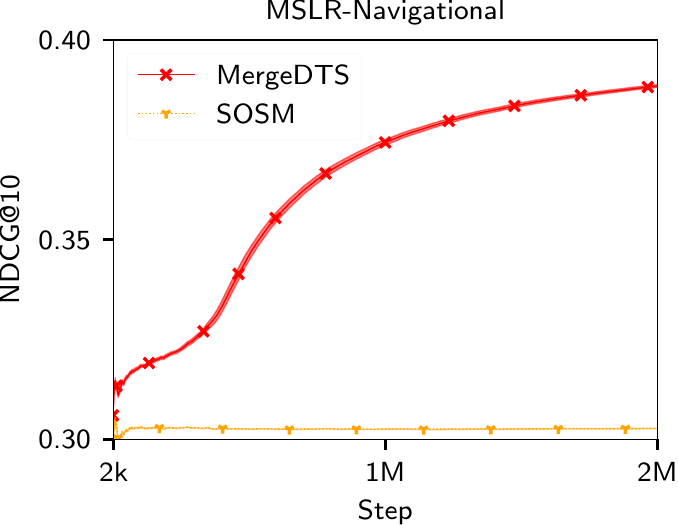}
	\vspace*{1mm}
	\includegraphics {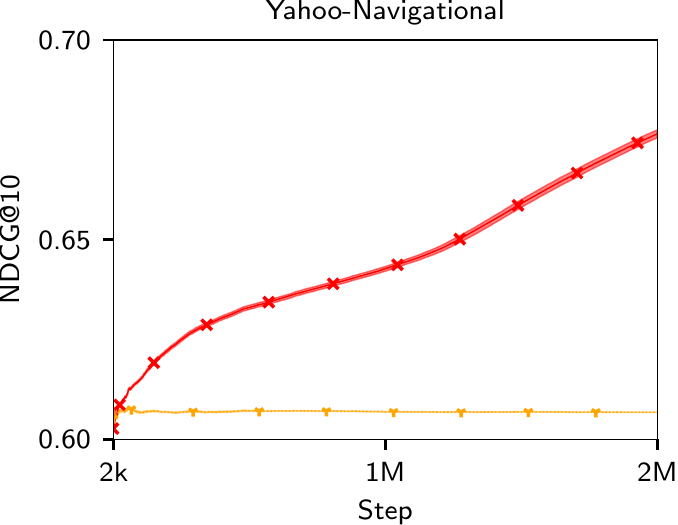}
	\includegraphics {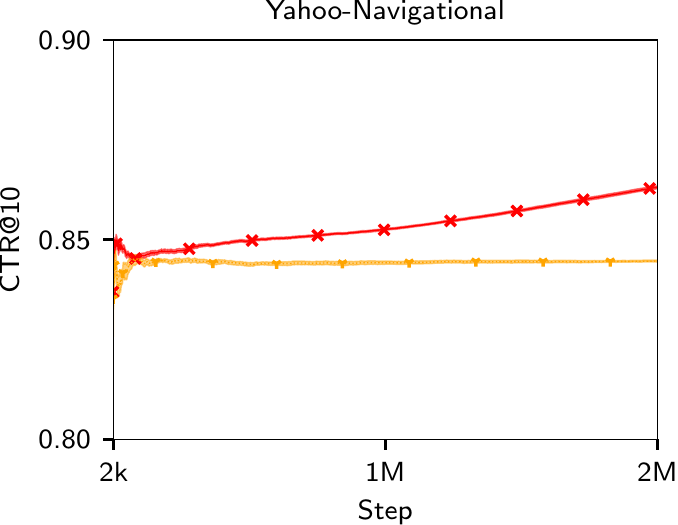}
	\vspace*{1mm}
	\includegraphics {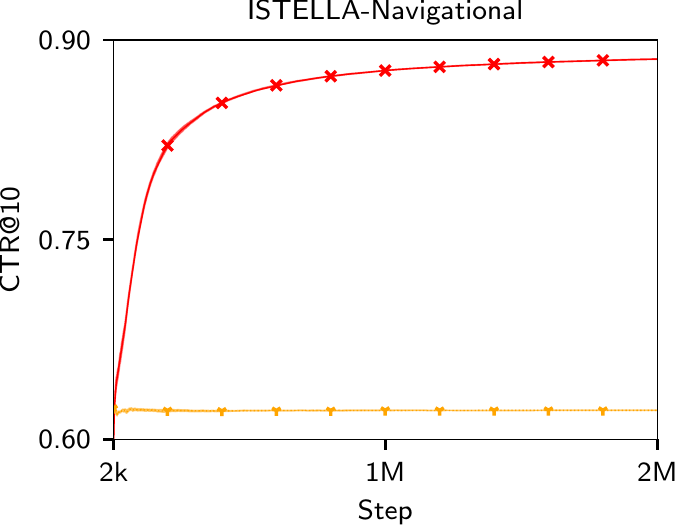}
	\includegraphics {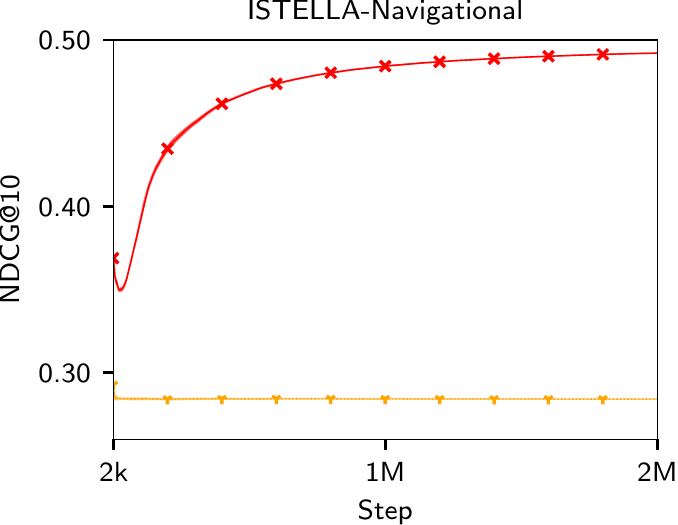}
	\vspace*{1mm}
	\caption{
			Left column: the average number of clicks received  at the top $10$ positions (higher is better). Right column: the average NDCG@10 of the interleaved or multileaved ranked lists (higher is better).  The results are averaged over $5$ independent runs. The shaded areas are $\pm$ standard error. 
	}
	\label{fig:multileave_ctr}
	\centering
\end{figure*}

To answer \textbf{RQ6}, we compare \ac{MergeDTS} together with \ac{PI}~\cite{hofmann-pi-2012} to \ac{SOSM}~\cite{Brost:2016},\footnote{The implementations of \ac{PI} and \ac{SOSM} are from this repository: \url{https://github.com/HarrieO/PairwisePreferenceMultileave}.} a multileaving method that is also designed for large-scale online ranker evaluation. 
Multileaving methods are designed to infer the preferences of rankers from multileaving comparisons. 
At each step, a multileaving method generates a ranked list from the lists produced by multiple rankers and infers the preferences of rankers based on the click feedback. 
Since the click feedback is noisy, the inferred preferences are also noisy. 
To get the best ranker, we need to run multileaving method a large number of steps. 
However, there is no multileaving method that estimates the confidence of the inferred preferences, thus we do not know when to stop the multileaving comparison.

In this experiment, we use the idea of \ac{UCB} to choose the rankers. 
Specifically, we run \ac{SOSM} as follows: 
\begin{inparaenum}[(1)]
	\item\label{item:1} 
	at each step, we compute \ac{UCB} estimators, defined in Table~\ref{tb:notation}, of the relative preferences of every pair of rankers;  
	\item \label{item:2}
	we consider rankers that have not been beaten by any other rankers based on the \ac{UCB} estimators; 
	\item \label{item:3}
	we use \ac{SOSM} to multileave the results of all considered rankers;
	\item \label{item:4}
	the multileaved list is shown to the user and click feedback is received; 
	\item\label{item:5} with the click feedback, we use \ac{SOSM} to infer the relative preferences of considered rankers and update the relative preferences to the matrix $\mathbf{W}$, defined in Table~\ref{tb:notation}.  
\end{inparaenum}
This approach is similar to \ac{MergeDTS} with only one batch. 
However, now we do not eliminate rankers. 
In Step~\ref{item:1}, to compute \ac{UCB} estimators, we choose $\alpha=0.51$. 
This is  a conservative setup because we may multileave the results of hundreds of rankers and the preferences inferred from a single multileaving comparison with hundreds of rankers can be rather noisy~\cite{Brost:2016}. 
For a fair comparison, in this experiment we also set $\alpha=0.51$ for \ac{MergeDTS}.  
Step~\ref{item:2} is from the Condorcet assumption that the Condorcet winner beats all others.

To conduct the multileaving experiments, we use the LEROT~\cite{lerot} online simulation setup instead of the proxy setup~\cite{mergerucb} used in other sections. 
We choose the MSLR, Yahoo  and ISTELLA datasets with the navigational configuration to simulate clicks and consider the top $10$ positions. 
These are standard setups in online ranker evaluation~\cite{Brost:2016,mergerucb,rucb,Schuth:2015:PMO:2766462.2767838,interleave,hofmann-reusing-2013}. 
Since we use simulated clicks in the experiments, we use the average number of clicks and the NDCG@10 of interleaved or multileaved results as metrics. 
We compute $NDCG@10$ of the ranked list $R$ as follows: 
	\begin{equation}
	NDCG@10(R) = \frac{DCG@10(R)}{DCG@10(R^*)}, \qquad DCG@10(R) = \sum_{i=1}^{10} \frac{2^{rel(R_i) }-1}{\log_2{(i+1)}}, 
	\end{equation}
where $R^*$ is the optimal ranking, $R_i$ is the $i$-th item in the ranked list $R$, and  $rel(R_i)$ is the relevance of the item $R_i$.  
Both metrics measure the quality of the displayed ranked lists. 
Higher values mean less loss of the user experience during online evaluation. 
The experiments are conducted for $2$M steps, with which \ac{MergeDTS} converges to the Condorcet winner, as shown in the above results. 
Since LETOR-based simulations are much slower than the proxy method~\cite{mergerucb}, the experiments are repeated (and averaged over) $5$ times. 

We report the number of clicks as CTR@$10$ in Figure~\ref{fig:multileave_ctr} and NDCG@10 in Figure~\ref{fig:multileave_ctr}.
\ac{MergeDTS}  outperforms \ac{SOSM} on both datasets. 
Particularly, the quality of the ranked lists identified by \ac{MergeDTS} increases quickly, which means that \ac{MergeDTS} can quickly eliminate suboptimal rankers. 
In the experiments, \ac{MergeDTS} finds the best ranker in  MSLR and ISTELLA datasets in less than $2$M steps. 
In the Yahoo dataset, \ac{MergeDTS} finds $2$ rankers in each repeat and in total $5$ repeats \ac{MergeDTS} finds $4$ rankers that contain the best ranker based on the NDCG@10. 
The curves of \ac{SOSM} are almost flat, which means that \ac{SOSM} always compares a large number of rankers. 
Specifically, in our experiments \ac{SOSM} compares all rankers in each step. 
Summarizing, the answer to \textbf{RQ6} is that \ac{MergeDTS} with \ac{PI} finds the best ranker faster than \ac{SOSM}.

\subsection{Parameter sensitivity }
\label{sec:parameter}

\begin{figure*}
	\centering
	\includegraphics {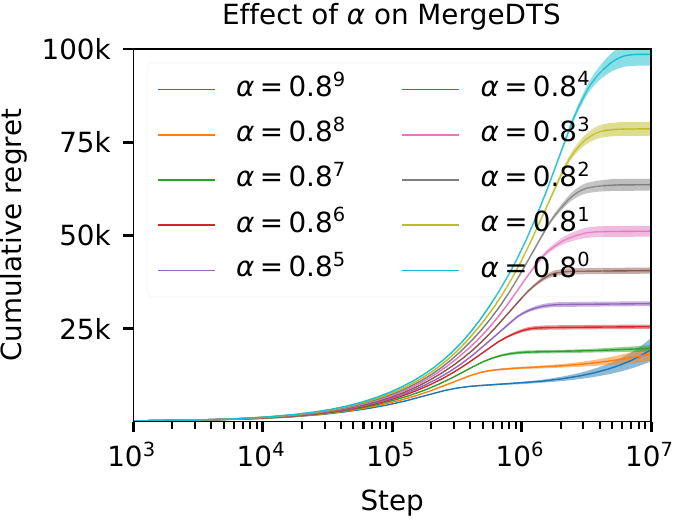}
	\includegraphics  {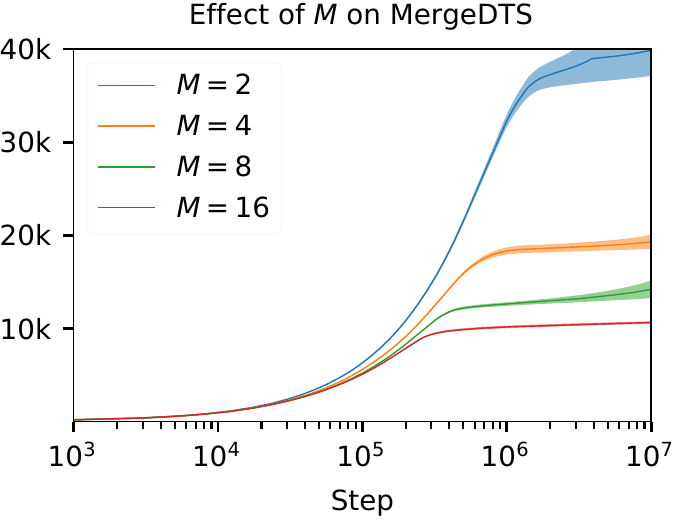}
	\includegraphics {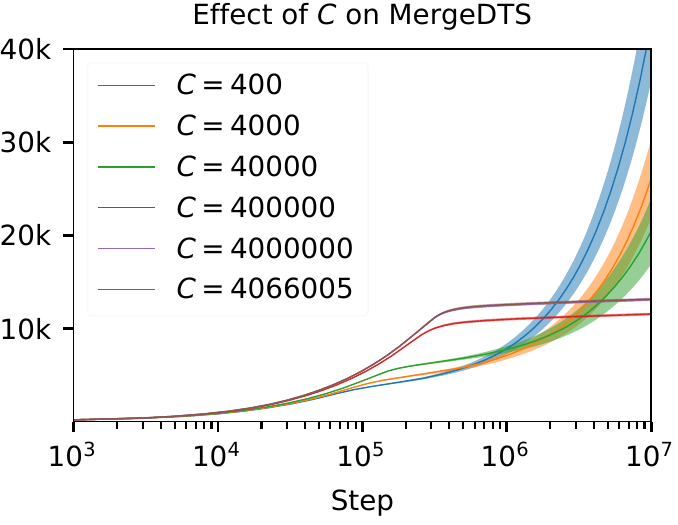}
	\caption{Effect of the parameters $\alpha$, $M$, and $C$ on the performance of \ac{MergeDTS} in the MSLR-Navigational setup. The results are averaged over $100$ independent runs. The shaded areas are $\pm$ standard error.}
	\label{fig:dependence}
\end{figure*}

We answer \textbf{RQ7} and analyze the parameter sensitivity of \ac{MergeDTS} using the setup described in Section~\ref{sec:parameters setup}. 
Since \ac{MergeDTS} converges to the Condorcet winner within $10$ million steps, we conduct the experiments with $10$ million steps and repeat $100$ times. 
Recall that we conduct the experiments in the MSLR-Navigational setup. 
The  results are reported in Figure~\ref{fig:dependence}. 
We also report the standard errors in the plots. 

The top-left plot in Figure~\ref{fig:dependence} shows the effect of the exploration parameter $\alpha$ on the performance of \ac{MergeDTS}.
First, lowering $\alpha$ can significantly increase the performance, e.g., the cumulative regret for $\alpha=0.8^4$ is about one third of the reward for $\alpha=1.0$ (which is close to the theoretically optimal value $\alpha=1.01$).
Second, as we decrease $\alpha$, the number of failures increases, where a failure is an event that \ac{MergeDTS} eliminates the Condorcet winner: 
with $\alpha=\{0.8^9, 0.8^8, 0.8^7\}$ we observe $10$, $4$, $1$ failures, respectively, and, thus, the cumulative regret increases linearly w.r.t.~$T$. 
Since in practice we do not want to eliminate the best ranker, we choose $\alpha = 0.8^6 \approx 0.2621$ in our experiments.

The top-right plot in Figure~\ref{fig:dependence} shows the effect of the batch size $M$.
The larger the batch size, the lower the regret. 
This can be explained as follows. 
The \ac{DTS}-based strategy uses the full local knowledge in a batch to choose the best ranker.  
A larger batch size $M$ provides more knowledge to \ac{MergeDTS} to make decisions, which leads to a better choice of rankers. 
But the time complexity of \ac{MergeDTS} is $O(TM^2)$, i.e., quadratic in the batch size.
Thus, for realistic scenarios we cannot increase $M$ indefinitely.
We choose $M=16$ as a tradeoff between effectiveness (cumulative regret) and efficiency (running time).

The bottom plot in Figure~\ref{fig:dependence} shows the dependency of \ac{MergeDTS} on $C$.
Similarly to the effect of $\alpha$, lower values of $C$ lead to lower regret, but also to a larger number of failures.
$C=\numprint{4000000}$ is the lowest value that does not lead to any failures, so we choose it in our experiments. 

In summary, the theoretical constraints on the parameters of MergeDTS are rather conservative. 
There is a range of values for the key parameters $\alpha$, $M$ and $C$, where the theoretical guarantees fail to hold, but where MergeDTS performs better than it would if we were to constrain ourselves only to values permitted by theory.
 
%\ch{
%	Let us briefly look back to the results in Sections~\ref{sec:large-scale} and~\ref{sec:noise}.
%	With the chosen parameters, \ac{MergeDTS} performs well in all nine setups. 
%	In all $50$ independent runs of each setup, \ac{MergeDTS} finds the Condorcet winner. 
%	From an empirical point of view, these are good results. 
%	Moreover, as the parameters are only tuned in the MSLR-Navigational setup, we believe that the performance of {MergeDTS} in other setups can be further improved by conducting a similar process in a corresponding setup.   
%	However, as Theorem~\ref{th:bound} does not hold with this parameter setup, MergeDTS in our experiments does not enjoy the theoretical guarantee.  
%}
%!TEX root = ../main.tex

\section{Conclusion}
\label{sec:conclusion}

In this paper, we have studied the large-scale online ranker evaluation problem under the Condorcet assumption, which can be formalized as a $K$-armed  dueling bandit problem.
We have proposed a scalable version of the state-of-the-art Double Thompson Sampling algorithm, which we call \ac{MergeDTS}. 

Our experiments have shown that, by choosing the parameter values outside of the theoretical regime, \ac{MergeDTS} is considerably more efficient than \ac{DTS} in terms of computational complexity, and that it significantly outperforms the large-scale state-of-the-art algorithm \ac{MergeRUCB}. 
Furthermore, we have demonstrated the robustness of \ac{MergeDTS} when dealing with difficult dueling bandit problems containing cycles among the arms. 
We have also demonstrated that \ac{MergeDTS} can be applied to some of the dueling bandit tasks which do not contain a Condorcet winner. 
Lastly, we have shown that the performance of \ac{MergeDTS} is guaranteed if the parameter values fall within the theoretical regime. 

Several interesting directions for future work arise from this paper:  
\begin{inparaenum}[(1)]
	\item 
		In our experiments, we have shown that there is a large gap between theory and practice. 
		It will be interesting to study this gap and provide a tighter theoretical bound. 
	\item We only study dueling bandits in this paper. 
	We believe that it is interesting to study a generalization of \ac{MergeDTS}, as well as the theoretical analysis presented here, to the case of online ranker evaluation tasks with a multi-dueling setup.
	\item Since multi-dueling bandits also compare multiple rankers at each step based on relative feedback, it is an interesting direction to compare dueling bandits to multi-dueling bandits in the large-scale setup.
	\item We suspect that the \ac{UCB}-based elimination utilized in \ac{MergeDTS} is too conservative, it might be that more recent minimum empirical divergence based techniques~\cite{cwrmed} may be leveraged to speed up the elimination of the rankers. 
	\item The feature rankers are chosen as arms in our experiments. 
	A more interesting and realistic way of choosing arms is to use well trained learning to rank algorithms. 
\end{inparaenum}

\section*{Code and data}
To facilitate reproducibility of the results in this paper, we are sharing the code and the data used to run the experiments in this paper at \url{https://github.com/chang-li/MergeDTS}.

\begin{acks}
We thank Artem Grotov, Rolf Jagerman, Harrie Oosterhuis, Christophe Van Gysel, and Nikos Voskarides for their helpful comments and technical support.

We also thank our editor and the anonymous reviewers for extensive comments and suggestions that helped us to improve the paper. 
\end{acks}

\bibliographystyle{ACM-Reference-Format}
\bibliography{references}

\end{document}